\documentclass[a4paper,11pt,onecolumn,nopdfoutputerror]{quantumarticle}
\usepackage[utf8]{inputenc}
\usepackage[english]{babel}
\usepackage[T1]{fontenc}
\usepackage{amsmath}
\usepackage{amsthm}
\usepackage{amsfonts}
\usepackage{amssymb}
\usepackage{hyperref}
\usepackage{braket}
\usepackage{booktabs}
\usepackage{array}
\usepackage[numbers,sort&compress]{natbib}
\usepackage{comment}

\newtheorem{theorem}{Theorem}
\newtheorem{lemma}[theorem]{Lemma}
\newtheorem{proposition}[theorem]{Proposition}
\newtheorem{definition}[theorem]{Definition}

\theoremstyle{definition}

\newtheorem{example}[theorem]{Example}

\newcommand{\F}{\mathbb F}
\newcommand{\Zn}{\mathbb Z_n}
\newcommand{\wt}{\operatorname{wt}}

\begin{document}

\title{Cyclic triorthogonal codes in prime dimension}

\author{Shiroman Prakash}
\affiliation{Department of Physics and Computer Science, Dayalbagh Educational Institute, Agra-282005, India}
\email{sprakash@dei.ac.in}

\maketitle

\begin{abstract}
We study cyclic triorthogonal codes in prime dimension $p$, on which a transversal
third-level diagonal gate acts as a logical non-Clifford gate. For odd $p$, we use a
definition of triorthogonality that does not require self-orthogonality. For cyclic codes
of length coprime to $p$, triorthogonality reduces to a sumset condition on the spectral
support, which makes exhaustive search feasible: for two natural constructions, we enumerate all maximal codes of this kind for $p=3$ up to length $121$ and $p=5,7$ up to length $96$. For qubits, we prove that a cyclic code of odd length is triply even if and only if it is classically triorthogonal, giving a complete description of triply-even cyclic codes, and we determine all such codes up to length $763$. For qudits, we find codes such as $[[19,1,3]]_3$ and $[[7,1,3]]_7$, which improve on the
shortest Reed--Muller and Reed--Solomon codes in both overhead exponent and distillation
threshold. The latter is the first member of an infinite family of
$[[p,1,2\lfloor p/6\rfloor+1]]_p$ codes.
\end{abstract}

\section{Introduction}

Triorthogonal codes~\cite{Bravyi_2012} are CSS codes for which a transversal third-level diagonal
gate implements a non-Clifford logical gate, and they underlie many magic state distillation
protocols~\cite{MSD,Bravyi_2012}. However, small explicit examples are comparatively scarce, particularly
for qudits. Many of the best-known examples are cyclic or extended cyclic. These include
the $[[15,1,3]]_2$ code~\cite{MSD}, the punctured Reed--Muller codes~\cite{CampbellAnwarBrowne,campbell2014enhanced,Hastings_2018,Haah_2018,Saha:2025frb}, and the qudit
Reed--Solomon codes~\cite{Krishna_2019}. 
It is therefore natural to study cyclic triorthogonal codes systematically.

Our study includes qudits of all prime dimensions. Magic state distillation extends naturally to qudits of odd prime
dimension~\cite{ACB,CampbellAnwarBrowne, nature}, where it can reduce overheads and raise thresholds
compared with qubits~\cite{campbell2014enhanced,DawkinsHoward,Gupta2020, Prakash:2024ghq, Prakash:2025azi}. The relevant
third-level gates and magic states, and the noise thresholds achievable with them, have been
studied in~\cite{HowardVala,jain2020qutrit,PhysRevA.83.032310}.

For cyclic codes, the algebra of triorthogonality becomes transparent. A cyclic code $C$ of
length $n$ over $\F_p$, with $p\nmid n$, is determined by its \emph{spectral support}
$S\subseteq\Zn$: the set of frequencies at which the discrete Fourier transform of a codeword
may be nonzero, or equivalently the exponents $j$ for which $\alpha^j$ is not a root of the
generator polynomial, where $\alpha$ is a primitive $n$-th root of unity in an extension field of $\F_p$. Coordinatewise
products of cyclic codes are again cyclic, and the span of all $t$-fold products of codewords
of $C$ has spectral support equal to the sumset $tS=S+\cdots+S$~\cite{Cascudo2019,CascudoGundersenRuano2020}.
Triorthogonality asks that certain coordinate sums of such products vanish, and a cyclic code
has vanishing coordinate sums exactly when $0$ is not in its support. Triorthogonality
therefore reduces to the condition that $0$ avoids certain sumsets of $S$. Camps-Moreno
et al.~\cite{CampsMoreno2024} used this idea to characterize binary cyclic CSS-T codes,
leaving the $q$-ary case open. We apply it to triorthogonality in every prime dimension.

We consider two constructions (Section~\ref{sec:cyclic}). In construction~A, the stabilizers
form a cyclic code and the logical row is the all-ones vector; this is the analogue of
punctured Reed--Muller codes, and triorthogonality is equivalent to $0\notin S\cup2S\cup3S$.
In construction~B, a cyclic space is punctured at one coordinate, and the
condition for triorthogonality is simply $0\notin3S$. For odd $p$ we use a definition of triorthogonality that,
unlike that of~\cite{Krishna_2019}, imposes no self-orthogonality, and we give a short
self-contained proof that it yields a transversal non-Clifford gate
(Section~\ref{sec:triortho}). T.~Saha has independently obtained this sufficient condition
for $p\ge5$~\cite{SahaNotes}. The weaker definition matters in practice: the $[[19,1,3]]_3$
code of Example~\ref{ex:19} is not self-orthogonal, and requiring self-orthogonality would
raise the shortest cyclic distance-$3$ qutrit code to length $39$.

\paragraph{Qubits.} For qubits the spectral criterion gives a complete description of a class of codes that is
otherwise hard to construct. A binary code is \emph{triply even} if all its weights are
divisible by $8$, and \emph{classically triorthogonal} if $\sum_lx_ly_lz_l\equiv0$ for all
codewords $x,y,z$. Every triply-even code is classically triorthogonal, but not
conversely~\cite{Betsumiya_2012}. We show that the converse holds for cyclic codes of odd
length: such a code is triply even if and only if $0\notin3S$, that is, if and only if it is
classically triorthogonal (Theorem~\ref{thm:te}). The proof is a short consequence of
McEliece's theorem on the divisibility of weights of cyclic codes~\cite{MCELIECE1972}. Systematic constructions of triply-even codes
are rare: beyond Reed--Muller codes, they have mostly come from doubling~\cite{Betsumiya_2012,Jain_2025}
and triangular graphs~\cite{Betsumiya_2012}. Here every union of cyclotomic cosets with
$0\notin3S$ gives one, and the transversal $T$ gate acts on the resulting quantum codes
exactly, with no Clifford correction. We determine all triply-even cyclic codes of odd length
up to $763$.

\paragraph{Qudits.} The spectral criteria make an exhaustive search over cyclic codes
arising from constructions A and B feasible, since each candidate is a union of cyclotomic cosets and the conditions can be
checked directly on $S$. We carry out this search for $p=3$ up to length $121$ and $p=5,7$
up to length $96$, and certify every distance, either exactly or as a rigorous lower bound
(Sections~\ref{sec:search} and~\ref{sec:results}). The resulting codes reach a given
distance at a fraction of the length of punctured Reed--Muller codes; for example, the new
$[[7,1,3]]_7$ code is less than a sixth of the length of $[[48,1,4]]_7$, the shortest
punctured Reed--Muller code with $d\ge3$. The criteria also yield infinite families. Alongside the qudit Reed--Solomon codes, we obtain
$[[p,1,2\lfloor p/6\rfloor+1]]_p$ codes for every prime $p\ge7$, of which $[[7,1,3]]_7$ is
the first (Proposition~\ref{prop:mds}); as far as we know, this family has not previously
been identified as triorthogonal. Finally, in Section~\ref{sec:perf} we compute exact
distillation thresholds and overhead exponents. The $[[19,1,3]]_3$ and $[[7,1,3]]_7$ codes improve on the shortest Reed--Muller and
Reed--Solomon codes in both; for $p=5$, however, cyclic codes provide no advantage over recursive use of the $[[4,1,2]]_5$
Reed--Solomon code.

\paragraph{Related work.} Zurel, Jana and de Silva~\cite{Zurel2026} recently showed that quantum
quadratic residue codes, which are also cyclic, unify several high-threshold distillation
protocols, including the $5$-qubit, Steane and $11$-qutrit Golay codes~\cite{2020golay}. Those
protocols distill Clifford-eigenstate magic states (the Bravyi--Kitaev $T$ state~\cite{MSD} for qubits and the
strange state~\cite{Veitch_2014} for qutrits) using codes with a transversal Clifford gate. The codes in this paper are of a
different kind: they distill the states $U\ket+$ produced by a transversal non-Clifford
gate $U$, which are the resource for injecting $U$.  For qubits, generalized triorthogonal codes have
recently been classified up to length $54$~\cite{WillsJainSingh2026}, extending~\cite{Nezami_2022}. The binary cyclic CSS-T characterization of~\cite{CampsMoreno2024} has been extended to
matrix-product codes with cyclic constituents~\cite{JaramilloVelez2026}. The authors of~\cite{CampsMoreno2024} also studied
binary triorthogonal codes within the CSS-T framework~\cite{Camps-Moreno:2024qcn}, and
proposed constructions from cyclic codes and the generalization to the $p$-ary case as
directions for future work. Binary CSS-T codes have also been built from cyclic codes~\cite{Berardini2024}, and triorthogonal
codes with cyclic column symmetry were found by integer programming in~\cite{Baldelli2026}, all of
distance~$3$ in the range searched. The shortest known triorthogonal codes at
distances $7$ and $9$, $[[95,1,7]]_2$ and $[[185,1,9]]_2$, come from non-cyclic
constructions~\cite{Sullivan_2024, Jain_2025}. Asymptotically good codes with transversal non-Clifford gates have also been
constructed~\cite{wills2024constantoverheadmagicstatedistillation,golowich2024asymptoticallygoodquantumcodes,nguyen2024goodbinaryquantumcodes,SanJose2026};
these target asymptotic overhead and typically require large lengths or field sizes, whereas we
focus on small codes for small primes.

The paper is organized as follows. Section~\ref{sec:prelim} reviews the spectral description of
cyclic codes. Section~\ref{sec:triortho} recalls the CSS code associated with a matrix, gives our
definition of triorthogonality for odd $p$, and proves that it yields a transversal gate,
separately for $p\ge5$ and $p=3$. Section~\ref{sec:cyclic} derives the spectral criteria, and
Section~\ref{sec:qubit} the qubit result. Sections~\ref{sec:search}--\ref{sec:perf} describe the
search, the codes and their distillation performance, and Section~\ref{sec:conclusion} concludes.
The appendices give the spectral description of concatenated codes, the proof of a lemma about distillation thresholds, the best code at every length, and the
spectral supports of selected codes.

\section{Preliminaries: Cyclic codes over $\F_p$ and spectral support}\label{sec:prelim}
\label{sec:spectral}

We collect the standard facts about cyclic codes that we need~\cite{macwilliams77}, in the
spectral language that makes products of codes easy to describe. Throughout this section
$\gcd(n,p)=1$.

Identify each vector $c\in\F_p^n$, with
coordinates $c_0,\dots,c_{n-1}$, with the polynomial $c(x)=\sum_lc_lx^l$ in the ring
$R_n=\F_p[x]/\langle x^n-1\rangle$. A cyclic shift
of the coordinates is multiplication by $x$, so a linear code is cyclic (invariant under cyclic
shifts) if and only if it is an ideal of $R_n$. Every ideal of $R_n$ is principal: a cyclic code
$C$ equals $\langle g(x)\rangle$ for a unique monic divisor $g$ of $x^n-1$, its \emph{generator
polynomial}, and $\dim C=n-\deg g$.

Let $m=\operatorname{ord}_n(p)$ be the least
$m$ with $n\mid p^m-1$, and let $K=\F_{p^m}$. Then $K^\times$, being cyclic of order $p^m-1$, contains
an element $\alpha$ of order exactly $n$, and
\[
x^n-1=\prod_{j\in\Zn}(x-\alpha^j),
\]
with distinct roots because $p\nmid n$. The \emph{$p$-cyclotomic coset} of $s\in\Zn$ is
$C_s=\{s,ps,p^2s,\dots\}\subseteq\Zn$. The Frobenius map $y\mapsto y^p$ permutes the roots
$\alpha^j$, sending $\alpha^j$ to $\alpha^{pj}$, so the polynomial $M_s(x)=\prod_{j\in C_s}(x-\alpha^j)$ is fixed
by Frobenius and has coefficients in $\F_p$; it is the minimal polynomial of $\alpha^s$ over $\F_p$.
Hence $x^n-1=\prod_sM_s(x)$, the product running over one representative of each coset, is the
factorization of $x^n-1$ into irreducibles over $\F_p$. Every generator polynomial is therefore a
product of some of the $M_s$, and cyclic codes of length $n$ correspond to unions of cosets.

For $c\in K^n$ define $\widehat c(j)=c(\alpha^j)=\sum_lc_l\alpha^{jl}$,
$j\in\Zn$. Since $\sum_{j\in\Zn}\alpha^{jl}=n\,\delta_{l,0}$ and $n$ is invertible in $\F_p$, the
transform is inverted by
\[
c_l=n^{-1}\sum_{j\in\Zn}\widehat c(j)\,\alpha^{-jl}.
\]

\begin{definition}
The \emph{spectral support} of a cyclic code $C=\langle g\rangle$ is $S(C)=\{j\in\Zn:g(\alpha^j)\ne0\}$.
\end{definition}

Equivalently, $\Zn\setminus S(C)$ is the set of zeros of $g$. Since $g$ is a product of minimal
polynomials, $S(C)$ is a union of cyclotomic cosets, and $|S(C)|=n-\deg g=\dim C$. A vector $c$ lies
in $C$ if and only if $g$ divides $c(x)$, i.e., if and only if $\widehat c(j)=0$ for every
$j\notin S(C)$.

It is convenient to think of this in terms of frequencies. For $s\in\Zn$, the \emph{pure frequency}
$e^{(s)}\in K^n$ is the vector with entries $e^{(s)}_l=n^{-1}\alpha^{-sl}$; its transform is
$\delta_s$. The inverse transform says that every vector is a superposition of pure frequencies,
$c=\sum_j\widehat c(j)\,e^{(j)}$, and a cyclic code consists of exactly the vectors that use only
frequencies in $S$. To make this literally true we allow coefficients in $K$: let $C_K$ denote the
$K$-linear span of $C$ in $K^n$. Then $C_K$ is the span of the $e^{(s)}$ with $s\in S$, since $C_K$
lies in that span and both have dimension $|S|$. (Working in $C_K$ is harmless: a linear condition,
such as the vanishing of a coordinate sum, holds on all of $C$ if and only if it holds on $C_K$,
because $C$ spans $C_K$.) For $t\ge1$ let
$C^{\langle t\rangle}=\operatorname{span}\{c^{(1)}*\cdots*c^{(t)}:c^{(i)}\in C\}$, which is again a cyclic code,
and $tS=S+\cdots+S$ ($t$ summands).

\begin{lemma}\label{lem:schur}
Let $p\nmid n$ and let $C$ be cyclic with spectral support $S$. Then:
\begin{enumerate}
\item $S(C^{\langle t\rangle})=tS$ for all $t\ge1$;
\item $\sum_lc_l=0$ for all $c\in C$ if and only if $0\notin S$;
\item $S(C^\perp)=\Zn\setminus(-S)$.
\end{enumerate}
\end{lemma}

\begin{proof}
All three parts follow from simple properties of pure frequencies.

(1) This is due to Cascudo~\cite{Cascudo2019}; see also~\cite{CascudoGundersenRuano2020,CampsMoreno2024}.
Multiplying pure frequencies adds their frequencies, since the exponents add:
$e^{(s_1)}*\cdots*e^{(s_t)}=n^{-(t-1)}e^{(s_1+\cdots+s_t)}$. As $C_K$ is spanned by the $e^{(s)}$,
$s\in S$, the products of $t$ of its elements span exactly the $e^{(k)}$ with $k\in tS$. These
products also span $(C^{\langle t\rangle})_K$, so $S(C^{\langle t\rangle})=tS$.

(2) The coordinate sum of a pure frequency is $\sum_le^{(s)}_l=n^{-1}\sum_l\alpha^{-sl}=\delta_{s,0}$.
So every vector of $C_K$, and hence of $C$, has coordinate sum $0$ if and only if $0\notin S$.

(3) Similarly, $\sum_le^{(s)}_le^{(s')}_l=n^{-1}\delta_{s+s',0}$: a pure frequency is orthogonal to every
other pure frequency except its negative. Writing $d=\sum_k\widehat d(k)e^{(k)}$, we get
$\sum_ld_le^{(s)}_l=n^{-1}\widehat d(-s)$, so $d$ is orthogonal to $C_K$ if and only if $\widehat d$
vanishes on $-S$. Hence $S(C^\perp)=\Zn\setminus(-S)$.
\end{proof}

For $p=2$ we have $S\subseteq2S\subseteq3S$: a coset is closed under doubling, so $s/2\in S$ and $s=s/2+s/2\in2S$, and $a+b=a+b/2+b/2\in3S$. For $p=3$ we have $S\subseteq3S$ but not $2S\subseteq3S$; for $p\ge5$ none of
these inclusions holds in general.

\begin{example}\label{ex:cyc8}
Consider $p=3$ and $n=8$, so that $m=2$ and $K=\F_9$. The cyclotomic cosets modulo $8$ are $\{0\}$, $\{1,3\}$,
$\{2,6\}$, $\{4\}$ and $\{5,7\}$. The code with $S=\{1,3\}$ has dimension $2$, and
\[
2S=\{2,4,6\},\qquad 3S=\{1,3,5,7\},
\]
since $1+1+1=3$, $1+1+3=5$, $1+3+3=7$ and $3+3+3=9\equiv1$. Neither contains $0$, so by
Proposition~\ref{prop:A} below this code, together with the all-ones logical row, is triorthogonal.
It gives an $[[8,1,2]]_3$ code. This is the punctured qutrit Reed--Muller code
of~\cite{CampbellAnwarBrowne}: punctured Reed--Muller codes are construction-A codes, and the search
finds a single maximal support at this length.
\end{example}

\section{Triorthogonal codes in odd prime dimension}\label{sec:triortho}

This section sets up the codes studied in the rest of the paper. We first recall the CSS code
associated with a matrix, for any prime $p$. We then define triorthogonality for odd $p$, and
show in Sections~\ref{sec:p5} and~\ref{sec:p3} that it guarantees a transversal non-Clifford gate.
The two cases are treated separately because for $p=3$ the relevant phases are ninth roots of
unity, which requires some care. For qubits we use the standard definition of Bravyi and
Haah~\cite{Bravyi_2012}, recalled in Section~\ref{sec:qubit}.

\subsection{The CSS code of a matrix}\label{sec:css}

Let $p$ be prime and let $H\in\F_p^{m\times n}$ have linearly independent rows $h^1,\dots,h^m$. We
call the first $k$ rows the \emph{logical rows}, forming $H_1$, and the remaining rows the
\emph{stabilizer rows}, forming $H_0$. Let $\mathcal A=\operatorname{rowspan}(H_0)$ and
$\mathcal H=\operatorname{rowspan}(H)$, and for $x\in\F_p^k$ write $x\cdot H_1=\sum_{i\le k}x_ih^i$.

\begin{lemma}\label{lem:css}
With no orthogonality assumption on $H$, the CSS code $Q(H)$ with $X$-stabilizers
$\{X(a):a\in\mathcal A\}$ and $Z$-stabilizers $\{Z(v):v\in\mathcal H^\perp\}$ is a valid code encoding
$k$ qudits, with code states $\ket{\bar x}=|\mathcal A|^{-1/2}\sum_{a\in\mathcal A}\ket{x\cdot H_1+a}$.
Logical operators can be taken as $\bar X_i=X(h^i)$ and $\bar Z_i=Z(z^i)$ with $z^i\in\mathcal A^\perp$
and $z^i\cdot h^j=\delta_{ij}$. The distances are
$d_Z=\min\{\wt(v):v\in\mathcal A^\perp\setminus\mathcal H^\perp\}$ and
$d_X=\min\{\wt(v):v\in\mathcal H\setminus\mathcal A\}$.
\end{lemma}\begin{proof}
$X(a)$ and $Z(v)$ commute because $a\cdot v=0$ for $a\in\mathcal A\subseteq\mathcal H$ and
$v\in\mathcal H^\perp$. States fixed by $Z(\mathcal H^\perp)$ are supported on $\mathcal H$,
and states fixed by $X(\mathcal A)$ are uniform on cosets of $\mathcal A$, so the code space
has dimension $|\mathcal H/\mathcal A|=p^k$. The $z^i$ exist because the rows of $H$ are
independent. Finally, $Z(v)$ commutes with the $X$-stabilizers if and only if
$v\in\mathcal A^\perp$, and such a $Z(v)$ acts trivially on the code space if and only if
$v\in\mathcal H^\perp$; similarly, $X(v)$ preserves the code space if and only if
$v\in\mathcal H$, and acts trivially on it if and only if $v\in\mathcal A$. This gives the
formulas for $d_Z$ and $d_X$.
\end{proof}

After twirling, the noise on the magic states considered here consists of $Z$-type errors only,
so the distance relevant for distillation is $d_Z$~\cite{Bravyi_2012,ACB,CampbellAnwarBrowne,Krishna_2019}.
Throughout, $[[n,k,d]]_p$ refers to $d=d_Z$; by Lemma~\ref{lem:dxdz} below, for triorthogonal codes
this is also the distance of the code.

\subsection{Triorthogonal matrices}\label{sec:tridef}

Let $p$ be odd. Triorthogonality is a condition on the trilinear form
$\sum_lx_ly_lz_l$ evaluated on the rows of a matrix. We allow nonzero weights
$\mathbf c\in(\F_p^\times)^n$ on the coordinates, and define
\[
\tau_{\mathbf c}(x,y,z)=\sum_l\mathbf c_lx_ly_lz_l\in\F_p\qquad(x,y,z\in\F_p^n),
\]
writing $\tau=\tau_{\mathbf 1}$. Non-unit weights are needed only in  Lemmas~\ref{lem:weights} and~\ref{lem:punct}, which we use in
Section~\ref{sec:cyclic} to describe construction~A as a puncturing and to discuss weighted
codes; unless stated otherwise, $\mathbf c=\mathbf 1$.

\begin{definition}\label{def:tri}
Let $H\in\F_p^{m\times n}$ have independent rows $h^1,\dots,h^m$, of which the first $k$
are logical rows and the rest stabilizer rows. The pair $(H,\mathbf c)$ is
\emph{triorthogonal} if
\begin{enumerate}
\item $\tau_{\mathbf c}(h^\alpha,h^\beta,h^\gamma)=0$ for all $\alpha,\beta,\gamma$, not
necessarily distinct, unless $\alpha=\beta=\gamma\le k$; and
\item $\lambda_i:=\tau_{\mathbf c}(h^i,h^i,h^i)\ne0$ for every logical row $i\le k$.
\end{enumerate}
\end{definition}

Condition~1 says that the form vanishes on every triple involving a stabilizer row, and
on every triple of logical rows that are not all equal; only a logical row repeated three
times may give a nonzero value, $\lambda_i$. Unlike the definition of~\cite{Krishna_2019},
no condition is imposed on the pairwise products $\sum_lh^\alpha_lh^\beta_l$; we return
to this point at the end of Section~\ref{sec:p3}.

We also use the corresponding notion for a subspace. A subspace $W\subseteq\F_p^N$ is a
\emph{triorthogonal space}, with weights $\mathbf c\in(\F_p^\times)^N$, if
$\tau_{\mathbf c}(x,y,z)=0$ for all $x,y,z\in W$. We write its length as $N$ because
such spaces are punctured to give codes of length $N-1$ (Lemma~\ref{lem:punct}).

The code associated with $(H,\mathbf c)$ is $Q(H)$ of Lemma~\ref{lem:css}, and its
transversal gate is $\bigotimes_lU^{\mathbf c_l}$, which applies $U^{\mathbf c_l}$ to
physical qudit $l$; here $U$ is the cubic gate defined in Section~\ref{sec:p5} for
$p\ge5$ and in Section~\ref{sec:p3} for $p=3$. We say that the transversal gate
\emph{implements} a gate $G$ on the $k$ logical qudits if it acts on the code states
$\ket{\bar x}$ as $G$ acts on $\ket{x_1,\dots,x_k}$. Sections~\ref{sec:p5}
and~\ref{sec:p3} show that if $(H,\mathbf c)$ is triorthogonal, the transversal gate
implements $\bigotimes_iU^{\lambda_i}$, which applies $U^{\lambda_i}$ to logical qudit $i$
(for $p=3$, up to Pauli corrections).

\begin{lemma}\label{lem:dxdz}
If $(H,\mathbf c)$ is triorthogonal, then $d_Z\le d_X$.
\end{lemma}
\begin{proof}
Let $u\in\mathcal H\setminus\mathcal A$ have weight $d_X$, and write
$u=\sum_\alpha x_\alpha h^\alpha$. Since $u\notin\mathcal A$, some logical coefficient is
nonzero, say $x_i\ne0$ with $i\le k$. Define $v\in\F_p^n$ by $v_l=\mathbf c_lu_l^2$.
Since $\mathbf c_l\ne0$, we have $v_l\ne0$ exactly when $u_l\ne0$, so $\wt(v)=d_X$. We
show that $v\in\mathcal A^\perp\setminus\mathcal H^\perp$, which gives $d_Z\le d_X$ by
Lemma~\ref{lem:css}.

For any row $h^\gamma$,
\[
v\cdot h^\gamma=\tau_{\mathbf c}(u,u,h^\gamma)
=\sum_{\alpha,\beta}x_\alpha x_\beta\,\tau_{\mathbf c}(h^\alpha,h^\beta,h^\gamma).
\]
If $h^\gamma$ is a stabilizer row, every term vanishes by Definition~\ref{def:tri}, so
$v\in\mathcal A^\perp$. If $\gamma=i$, the only term that can survive is
$\alpha=\beta=i$, so $v\cdot h^i=\lambda_ix_i^2\ne0$ and $v\notin\mathcal H^\perp$.
\end{proof}

The same argument applies to qubits with $v=u$: in a Bravyi--Haah triorthogonal matrix,
distinct rows are orthogonal, so $u\cdot h^\gamma=x_\gamma\wt(h^\gamma)\bmod2$, which
vanishes for stabilizer rows and equals $x_i=1$ for a logical row with $x_i\ne0$.

\subsection{The case $p\ge5$}\label{sec:p5}

Let $p\ge5$ and $\omega=e^{2\pi i/p}$, and define the cubic gate
\[
U=\sum_{j\in\F_p}\omega^{j^3}\ket j\bra j ,
\]
the third-level gate $U_{1,3}$ of~\cite{Cui,Krishna_2019}. Its phases are $p$-th roots of unity and
depend only on $j^3\in\F_p$, so all arithmetic can be done in $\F_p$.

\begin{theorem}\label{thm:conv}
Let $p\ge5$ and let $(H,\mathbf c)$ be triorthogonal. Then $\bigotimes_lU^{\mathbf c_l}$ maps $\ket{\bar x}$ to
$\omega^{\sum_i\lambda_ix_i^3}\ket{\bar x}$; that is, it implements $\bigotimes_iU^{\lambda_i}$ exactly, with no
correction.
\end{theorem}
\begin{proof}
By Lemma~\ref{lem:css}, $\ket{\bar x}$ is a uniform superposition of the basis states $\ket f$ with
\[
f=\sum_{\alpha=1}^mx_\alpha h^\alpha\in\F_p^n,
\]
where $x_1,\dots,x_k$ are the logical values and $x_{k+1},\dots,x_m\in\F_p$ range over all stabilizer
coefficients. The transversal gate multiplies $\ket f$ by $\omega^E$ with $E=\sum_l\mathbf c_lf_l^3\in\F_p$,
and we must show that $E$ depends on $x_1,\dots,x_k$ only. Since $f_l^3=f_lf_lf_l$, we have
$E=\tau_{\mathbf c}(f,f,f)$, and by trilinearity
\[
E=\sum_{\alpha,\beta,\gamma=1}^mx_\alpha x_\beta x_\gamma\,\tau_{\mathbf c}(h^\alpha,h^\beta,h^\gamma),
\]
the sum running over ordered triples. By Definition~\ref{def:tri} every term vanishes except those
with $\alpha=\beta=\gamma\le k$, so $E=\sum_{i\le k}\lambda_ix_i^3$.
\end{proof}

\subsection{The case $p=3$}\label{sec:p3}

For qutrits the map $j\mapsto j^3$ is the identity on $\F_3$, so $\omega^{j^3}$ is a Pauli. The
non-Clifford gate instead has phases that are ninth roots of unity, and the proof must pass between
$\F_3$ and the integers.

An integer $s$ is a \emph{lift} of $a\in\F_p$ if $s\equiv a\pmod p$. The
\emph{standard lift} $\iota:\F_p\to\{0,1,\dots,p-1\}\subset\mathbb Z$ sends $a$ to its unique lift in that
range, and lifts of vectors are taken entrywise. Lifts respect arithmetic only modulo $p$: if $s$ and
$t$ lift $a$ and $b$, then $s+t$ and $st$ lift $a+b$ and $ab$, but in general
$\iota(a)+\iota(b)\ne\iota(a+b)$. For example, in $\F_3$ we have $2+2=1$, yet
$\iota(2)+\iota(2)=4\ne1=\iota(1)$; the two lifts differ by a multiple of $3$ (a ``carry'').

Let $\omega=e^{2\pi i/3}$ and $\zeta=e^{2\pi i/9}$, and define
\[
U=\sum_{j\in\F_3}\zeta^{\iota(j)^3}\ket j\bra j=\operatorname{diag}(1,\zeta,\zeta^{-1}),
\]
the gate of~\cite{HowardVala}; the last form uses $\iota(2)^3=8\equiv-1\pmod9$. The proof rests on the
fact that carries do not affect cubes modulo $9$.

\begin{lemma}\label{lem:cubelift}
If $s$ and $s'$ are integers with $s\equiv s'\pmod3$, then $s^3\equiv s'^3\pmod9$. Consequently
$U\ket j=\zeta^{s^3}\ket j$ for \emph{every} lift $s$ of $j$, not only for $s=\iota(j)$.
\end{lemma}
\begin{proof}
Write $s'=s+3t$. Then $s'^3-s^3=9s^2t+27st^2+27t^3$, which is divisible by $9$.
\end{proof}

For instance, $4$ lifts $1\in\F_3$, and indeed $4^3=64\equiv1=\iota(1)^3\pmod9$. We will use two
further facts. Firstly, $s^3\equiv s\pmod3$ for every integer $s$, by Fermat's little theorem, so
$3as^3\equiv3as\pmod9$ for all integers $a$ and $s$. Secondly,
$U^3=\sum_j\zeta^{3\iota(j)^3}\ket j\bra j=\sum_j\omega^{\iota(j)}\ket j\bra j=Z$, so powers of $U$ whose
exponents agree modulo $3$ differ by a Pauli, namely $U^{a+3t}=U^aZ^t$.
In particular, for a weight $\mathbf c_l\in\F_3^\times$ we take $U^{\mathbf c_l}$ to mean $U^{\tilde{\mathbf c}_l}$ for a lift
$\tilde{\mathbf c}_l$ of $\mathbf c_l$; different lifts give the same gate up to a Pauli.

\begin{theorem}\label{thm:conv3}
Let $p=3$ and let $(H,\mathbf c)$ be triorthogonal. Then there is $w\in\F_3^n$ such that
$Z(w)\bigotimes_lU^{\tilde{\mathbf c}_l}$ implements $\bigotimes_iU^{\Lambda_i}$, for integers
$\Lambda_i\equiv\lambda_i\pmod3$. Thus the transversal gate implements $\bigotimes_iU^{\lambda_i}$ up to a fixed
Pauli correction and a logical Pauli.
\end{theorem}
\begin{proof}
As in the proof of Theorem~\ref{thm:conv}, $\ket{\bar x}$ is a uniform superposition of the states
$\ket f$ with $f=\sum_\alpha x_\alpha h^\alpha\in\F_3^n$, and the transversal gate multiplies $\ket f$ by $\zeta^E$
with $E=\sum_l\tilde{\mathbf c}_l\,\iota(f_l)^3\bmod9$. We compute $E$ in five steps.

\emph{Step 1: an integer lift of $f$.} Define the integer vector
$\tilde f=\sum_\alpha\iota(x_\alpha)\,\iota(h^\alpha)\in\mathbb Z^n$, computed in $\mathbb Z$ with no reduction.
Reducing modulo $3$ gives back $f$, so $\tilde f_l$ is a lift of $f_l$ for every $l$. By
Lemma~\ref{lem:cubelift}, $\iota(f_l)^3\equiv\tilde f_l^{\,3}\pmod9$, hence
$E\equiv\sum_l\tilde{\mathbf c}_l\tilde f_l^{\,3}\pmod9$. Working with $\tilde f$ instead of $\iota(f)$ is the
key step: $\tilde f$ is linear in the integers $\iota(x_\alpha)$, whereas $\iota(f)$ is not, because of
carries.

\emph{Step 2: expansion.} Expanding the cube over ordered triples and grouping them by how many
indices coincide,
\begin{multline*}
E\equiv\sum_\alpha\iota(x_\alpha)^3\,\Lambda_\alpha
+3\sum_{\alpha\ne\beta}\iota(x_\alpha)^2\iota(x_\beta)\,\tilde\tau_{\alpha\alpha\beta}\\
+6\sum_{\alpha<\beta<\gamma}\iota(x_\alpha)\iota(x_\beta)\iota(x_\gamma)\,\tilde\tau_{\alpha\beta\gamma}\pmod9,
\end{multline*}
where the integers
\[
\tilde\tau_{\alpha\beta\gamma}=\sum_l\tilde{\mathbf c}_l\,\iota(h^\alpha_l)\iota(h^\beta_l)\iota(h^\gamma_l),\qquad
\Lambda_\alpha=\tilde\tau_{\alpha\alpha\alpha},
\]
are lifts of $\tau_{\mathbf c}(h^\alpha,h^\beta,h^\gamma)$ and of $\tau_{\mathbf c}(h^\alpha,h^\alpha,h^\alpha)$. The
coefficients $3$ and $6$ count the orderings of a triple with exactly two, or no, equal indices.

\emph{Step 3: the mixed terms vanish.} In the second and third sums every term carries a factor $3$,
so modulo $9$ it depends only on the remaining factors modulo $3$; in particular
$\tilde\tau_{\alpha\alpha\beta}$ and $\tilde\tau_{\alpha\beta\gamma}$ enter only through their residues
$\tau_{\mathbf c}(h^\alpha,h^\alpha,h^\beta)$ and $\tau_{\mathbf c}(h^\alpha,h^\beta,h^\gamma)$. These vanish by
Definition~\ref{def:tri}, since the indices are not all equal. So both sums are $\equiv0\pmod9$.

\emph{Step 4: stabilizer terms are Pauli phases.} For a stabilizer row $\alpha>k$,
$$\Lambda_\alpha\equiv\tau_{\mathbf c}(h^\alpha,h^\alpha,h^\alpha)=0\pmod3,$$ so $\Lambda_\alpha=3s_\alpha\bmod9$ for an
integer $s_\alpha$. By Fermat's little theorem, $3s_\alpha\iota(x_\alpha)^3\equiv3s_\alpha\iota(x_\alpha)\pmod9$, so the stabilizer
terms contribute the phase \begin{equation*}
\zeta^{3\sum_{\alpha>k}s_\alpha\iota(x_\alpha)}=\omega^{\sum_{\alpha>k}s_\alpha x_\alpha},
\end{equation*} which is
\textit{linear} in the stabilizer coefficients. Because the rows of $H$ are independent, there is
$w\in\F_3^n$ with $w\cdot h^\alpha=-s_\alpha$ for $\alpha>k$ and $w\cdot h^i=0$ for $i\le k$. The Pauli operator $Z(w)$
multiplies $\ket f$ by $\omega^{w\cdot f}=\omega^{-\sum_{\alpha>k}s_\alpha x_\alpha}$, cancelling these terms exactly and
leaving the logical part unchanged.

\emph{Step 5: the logical gate.} We are left with $\zeta^{\sum_{i\le k}\Lambda_i\iota(x_i)^3}$ which is the phase of
$\bigotimes_iU^{\Lambda_i}$ on $\ket{x_1,\dots,x_k}$. Finally $\Lambda_i\equiv\lambda_i\pmod3$, so
$U^{\Lambda_i}=U^{\lambda_i}Z^{t_i}$ for some integer $t_i$, a logical Pauli.
\end{proof}

\paragraph{Remark on pairwise conditions.} No pairwise condition $\sum_lh^a_lh^b_l=0$ is used in
Definition~\ref{def:tri} or in the proofs above. The definition of~\cite{Krishna_2019} includes such
conditions partly because, for $p=3$, it uses the gate $V=\operatorname{diag}(1,\zeta,\zeta^2)$. One can
check that $V=UD$ with $D=\operatorname{diag}(1,1,\omega)$, a diagonal Clifford gate. Hence the magic
states $U\ket+$ and $V\ket+$ are Clifford-equivalent. However, $D^{\otimes n}$ contributes a phase
quadratic in the stabilizer coefficients, which a Pauli can cancel only if the stabilizers are
orthogonal to every row. With $U$ these conditions are not needed. A fuller treatment, including
converses of Theorems~\ref{thm:conv} and~\ref{thm:conv3}, will be given elsewhere.

\subsection{Weights and puncturing}

We also note two simple facts
that will be used in Section~\ref{sec:cyclic}.

\begin{lemma}\label{lem:weights}
$W$ admits weights $\mathbf c$ making it a triorthogonal space if and only if
$(W^{\langle3\rangle})^\perp$ contains a vector with no zero entries.
\end{lemma}
\begin{proof}
$\tau_{\mathbf c}(x,y,z)=\langle\mathbf c,x*y*z\rangle$, so the condition is $\mathbf c\perp W^{\langle3\rangle}$.
\end{proof}

\begin{lemma}[Puncturing]\label{lem:punct}
Let $W\subseteq\F_p^N$ be a triorthogonal space with weights $\mathbf c$, and $s$ a coordinate on which
$W$ is not identically zero. Choose $g\in W$ with $g_s=1$ and a basis of $\{w\in W:w_s=0\}$, and
delete coordinate $s$. The result is a triorthogonal matrix with one logical row and
$\lambda=-\mathbf c_s$, and its distance is
\[
d_Z=\min\{\wt(u)-1:u\in W^\perp,\ u_s\ne0\}.
\]
\end{lemma}\begin{proof}
Write $W_0=\{w\in W:w_s=0\}$, so that $W=\langle g\rangle\oplus W_0$, and for
$x\in\F_p^N$ let $x'$ be $x$ with coordinate $s$ deleted. The punctured matrix has logical
row $g'$, stabilizer space $\mathcal A=W_0'$, and $\mathcal H=W'$.

\emph{Triorthogonality.} For $x,y,z\in W$, the form on the remaining coordinates is
\[
\tau_{\mathbf c}(x,y,z)-\mathbf c_sx_sy_sz_s=-\mathbf c_sx_sy_sz_s,
\]
since $\tau_{\mathbf c}$ vanishes on $W$. This is zero if any of $x,y,z$ lies in $W_0$, and
equals $-\mathbf c_s$ for $x=y=z=g$. So Definition~\ref{def:tri} holds with
$\lambda=-\mathbf c_s$.

\emph{Independence.} If the rows became dependent after deletion, some nonzero $w\in W$
would have $w'=0$, that is, $w=ae_s$ with $a\ne0$. Then
$\tau_{\mathbf c}(w,w,w)=a^3\mathbf c_s\ne0$, contradicting the triorthogonality of $W$.

\emph{Distance.} Given $v\in\F_p^{N-1}$, let $u$ be $v$ with $t=-v\cdot g'$ inserted in
position $s$. Then $u\cdot g=0$, and $u\cdot w=v\cdot w'$ for $w\in W_0$, so
$v\in\mathcal A^\perp$ if and only if $u\in W^\perp$. Moreover $v\in\mathcal H^\perp$ if and
only if in addition $v\cdot g'=0$, that is, $t=0$. Hence $v\mapsto u$ is a bijection from
$\mathcal A^\perp\setminus\mathcal H^\perp$ onto $\{u\in W^\perp:u_s\ne0\}$, with
$\wt(v)=\wt(u)-1$.
\end{proof}
\section{Cyclic constructions}\label{sec:cyclic}

We now give two ways of building triorthogonal codes from cyclic codes. In both,
triorthogonality becomes a sumset condition on $S$. Throughout, $p\nmid n$, $C$ and $W$ are
cyclic codes with spectral support $S$, and triorthogonality is unweighted unless stated
otherwise.

\subsection{Two constructions}

Both constructions are obtained by puncturing a triorthogonal space
(Lemma~\ref{lem:punct}). In construction~B the space is a cyclic code $W$, and the punctured
quantum code is not cyclic. In construction~A the space is spanned by $(\mathbf 1\,|\,1)$ and
$(C\,|\,0)$, with weight $-n$ on the extra coordinate; puncturing it at that coordinate gives
the cyclic quantum code with stabilizers $C$ and logical row $\mathbf 1$. Construction~A is
thus formally a weighted case of construction~B, but its unpunctured space is extended cyclic
rather than cyclic, so we treat it separately. It is the analogue of punctured Reed--Muller
codes.

\begin{proposition}[Construction A: cyclic stabilizers, all-ones logical]\label{prop:A}
Let $H$ have logical row $\mathbf 1=(1,\dots,1)$ and stabilizer rows a basis of $C$. Then $H$ is
triorthogonal (in the sense of Bravyi--Haah for $p=2$, and of Definition~\ref{def:tri} with
$\mathbf c=\mathbf 1$ for odd $p$) if and only if
\[
0\notin S,\qquad0\notin2S,\qquad0\notin3S .
\]
Here $\lambda=n\bmod p$, which is nonzero since $p\nmid n$. For $p=2$ the condition reduces to $0\notin3S$, and for $p=3$ to
$0\notin2S\cup3S$. The distance is
\[
d=\min\{\wt(v): v\in C^\perp,\ \textstyle\sum_lv_l\ne0\},
\]
where $C^\perp$ is cyclic with spectral support $\Zn\setminus(-S)$.
\end{proposition}\begin{proof}
By trilinearity, the conditions of Definition~\ref{def:tri} on triples involving a
stabilizer row hold for the basis of $C$ if and only if they hold for all $x,y,z\in C$.
With logical row $\mathbf 1$, these triples give
\[
\tau(\mathbf1,\mathbf1,x)=\textstyle\sum_lx_l,\qquad
\tau(\mathbf1,x,y)=\sum_lx_ly_l,\qquad
\tau(x,y,z)=\sum_lx_ly_lz_l,
\]
so triorthogonality says exactly that $C$, $C^{\langle2\rangle}$ and $C^{\langle3\rangle}$
have zero coordinate sums, which by Lemma~\ref{lem:schur} is
$0\notin S\cup2S\cup3S$. The remaining condition is
$\lambda=\tau(\mathbf1,\mathbf1,\mathbf1)=n\ne0$. For $p=2$, the Bravyi--Haah conditions
say that the stabilizer rows have even weight and that every product of two or three
distinct rows has even weight; since $x*x=x$ over $\F_2$, this is again the statement that
$C$, $C^{\langle2\rangle}$ and $C^{\langle3\rangle}$ have even weights. For the distance,
$\mathcal A^\perp=C^\perp$ and $\mathcal H^\perp=C^\perp\cap\mathbf1^\perp$, so
$\mathcal A^\perp\setminus\mathcal H^\perp$ consists of the vectors of $C^\perp$ with
nonzero coordinate sum.
\end{proof}
\begin{proposition}[Construction B: punctured cyclic spaces]\label{prop:B}
$W$ is an (unweighted) triorthogonal space if and only if $0\notin3S$. Puncturing a nonzero such $W$
at one coordinate gives an $[[n-1,1,d]]_p$ triorthogonal code with $\lambda=-1$ and
\[
d=d(W^\perp)-1 .
\]
\end{proposition}
\begin{proof}
Since $\tau(x,y,z)=\langle\mathbf1,x*y*z\rangle$, $\tau$ vanishes on $W$ if and only if
$W^{\langle3\rangle}$ has zero coordinate sums, that is, $0\notin S(W^{\langle3\rangle})=3S$,
by Lemma~\ref{lem:schur}. A nonzero cyclic code is nonzero at every coordinate, so
Lemma~\ref{lem:punct} applies at any coordinate $s$, with $\mathbf c_s=1$ and hence
$\lambda=-1$. Since $W^\perp$ is also cyclic, any minimum-weight vector of $W^\perp$ can be
shifted so that $u_s\ne0$, so the distance given by Lemma~\ref{lem:punct} is
$d(W^\perp)-1$.
\end{proof}

The two criteria differ only in the conditions on $S$ and $2S$. For $p=2$ these follow from
$0\notin3S$, since $S\subseteq2S\subseteq3S$, so both constructions use the same supports. For
$p=3$ we have $S\subseteq3S$, so construction~A imposes exactly one extra condition,
$0\notin2S$; Example~\ref{ex:19} shows that this matters. For $p\ge5$ both extra conditions
are meaningful. For $p=2$, Theorem~\ref{thm:te} shows that $W$ is triply even, so after
puncturing every component of $\ket{\bar0}$ has weight divisible by $8$ and every component
of $\ket{\bar1}$ has weight $\equiv7\pmod8$; hence $T^{\otimes(n-1)}$ implements $T^\dagger$
exactly.

\begin{example}\label{ex:19}
Consider the ternary cyclic code of length $20$ with generator polynomial
$$2+2x+x^2+x^3+2x^4+2x^6+x^8+x^{10}+2x^{11}+2x^{12}+x^{13}+x^{14}.$$
Its spectral support is $S=C_2\cup C_5$, with $C_2=\{2,6,14,18\}$ and $C_5=\{5,15\}$. Since
$5+15\equiv0$, we have $0\in2S$, so the resulting code is not self-orthogonal; but $0\notin3S$, so
construction~B applies. Puncturing once gives the $[[19,1,3]]_3$ code with
\[
H=\left(\begin{matrix}
0&0&0&0&0&2&2&1&1&2&0&2&0&1&0&1&2&2&1\\ \hline
1&0&0&0&0&1&0&1&0&2&2&1&2&2&1&2&2&0&1\\
0&1&0&0&0&1&2&2&0&1&2&0&1&1&2&0&0&0&2\\
0&0&1&0&0&2&0&0&0&2&1&1&0&2&1&0&2&2&1\\
0&0&0&1&0&1&0&2&2&1&2&2&1&2&2&0&1&0&1\\
0&0&0&0&1&1&2&2&1&0&1&0&2&0&2&1&1&2&2
\end{matrix}\right),
\]
with $\lambda=-1$; the row above the line is the logical row. We verified exhaustively that
$U^{\otimes19}$, followed by a Pauli, implements logical $U^{-1}$.
Computing its distance via brute force gives $d=3$.
\end{example}

\subsection{Two MDS families}

Two infinite families follow directly from the spectral criteria. Both use the BCH bound: if the
zeros of a cyclic code include $\delta-1$ consecutive frequencies $b,b+1,\dots,b+\delta-2$, then
every nonzero codeword has weight at least $\delta$~\cite{macwilliams77}. Each family is MDS in the
sense that the relevant dual code meets the Singleton bound.

\begin{proposition}[Two MDS families]\label{prop:mds}
Let $p$ be an odd prime and $r\ge1$; the first family is nonempty for $p\ge5$, and the second for
$p\ge7$.
\begin{enumerate}
\item Construction A with $n=p-1$ and $S=\{1,\dots,r\}$, where $3r<p-1$, gives a
$[[p-1,1,r+1]]_p$ code. With $r=\lfloor(p-2)/3\rfloor$ these are the Reed--Solomon
codes~\cite{campbell2014enhanced,Krishna_2019}.
\item Construction B with $n=p+1$, $\nu=n/2$ and $S=\{\nu-r,\dots,\nu+r\}$, where $6r<p+1$, gives a
$[[p,1,2r+1]]_p$ code. With $r=\lfloor p/6\rfloor$ this is a $[[p,1,2\lfloor p/6\rfloor+1]]_p$ code.
\end{enumerate}
In both cases $\lambda=-1$.
\end{proposition}
\begin{proof}
(1) Since $p\equiv1\pmod n$, every cyclotomic coset is a singleton. The sets $S$, $2S$ and $3S$ all
lie in $\{1,\dots,3r\}$, and $3r<n$, so none of them contains $0$. The code $C^\perp$ has the $r$
consecutive zeros $-S$, so $d(C^\perp)\ge r+1$ by the BCH bound, with equality by the Singleton
bound since $\dim C^\perp=n-r$. The subcode $C^\perp\cap\mathbf 1^\perp$ has the $r+1$ consecutive
zeros $-S\cup\{0\}$, hence minimum distance at least $r+2$, so every minimum-weight vector of
$C^\perp$ has nonzero coordinate sum, and $d=r+1$ by Proposition~\ref{prop:A}. Finally
$\lambda=n\equiv-1$.

(2) Since $p\equiv-1\pmod n$, every cyclotomic coset is $\{j,-j\}$. As $2\nu\equiv0$, we have
$-(\nu+t)\equiv\nu-t$, so $S$ is a union of cosets, and $3S=\{\nu-3r,\dots,\nu+3r\}$ (using
$3\nu\equiv\nu$), which avoids $0$ if and only if $3r<\nu$. The code $W^\perp$ has the $2r+1$
consecutive zeros $-S=S$, so $d(W^\perp)=2r+2$ by the BCH and Singleton bounds, and $d=2r+1$ by
Proposition~\ref{prop:B}.
\end{proof}

The $[[7,1,3]]_7$ code of Section~\ref{sec:results} is the smallest member of the second
family; its support $\{1,4,7\}$ is the image of $\{3,4,5\}$ under the multiplier $3^{-1}$.
In the first family $C^\perp$ is a Reed--Solomon code on the multiplicative group of $\F_p$,
and in the second $W^\perp$ is a cyclic MDS code of length $p+1$. As far as we know, the
second family has not previously been identified as a family of triorthogonal codes. This
does not follow from classical equivalences, which in general rescale coordinates, and
rescaling does not preserve unweighted triorthogonality.

The two families compare as follows. For $p\equiv1\pmod6$ the second family has
$d=(p+2)/3$ at length $p$, against $d=(p-1)/3$ at length $p-1$ for the Reed--Solomon codes,
so it gains one unit of distance for one extra qudit (for example, $[[7,1,3]]_7$ against
$[[6,1,2]]_7$, and $[[13,1,5]]_{13}$ against $[[12,1,4]]_{13}$). For $p\equiv5\pmod6$ these
become $(p-2)/3$ and $(p+1)/3$, and the Reed--Solomon codes are better. We verified both
families, including distances and transversal conditions, for $p\le19$.

\subsection{Weights and concatenation}

We consider only unit weights. By Lemma~\ref{lem:weights}, a cyclic
space $W$ can be made triorthogonal with some weights $\mathbf c$ if and only if the cyclic
code $(W^{\langle3\rangle})^\perp$, whose spectral support is $\Zn\setminus(-3S)$, contains a
vector with no zero entries; the unweighted case corresponds to $0\notin3S$. Weights that are
cubes, $\mathbf c_l=u_l^3$, give nothing essentially new: rescaling coordinate $l$ by $u_l$
turns the weighted code into an unweighted one with the same distance and logical gate,
although it is in general no longer cyclic. For $p=3$ and $p=5$ every unit is a cube, so
weights can only produce such rescaled codes; for $p\equiv1\pmod3$ they can also change the
transversal gate. For $p=3$ and $n\le44$, every cyclic space that is triorthogonal for some nonconstant
weights but not for $\mathbf c=\mathbf 1$ gives $d=1$, but we do not know whether this
holds in general.

Concatenations of construction-A codes are again construction-A codes, provided the blocks
are interleaved, so that block $b$ consists of the coordinates congruent to $b$ modulo the
outer length. Concatenating an inner code with support $S_1\subseteq\mathbb Z_{k_1}$ with an
outer code with support $S_2\subseteq\mathbb Z_{k_2}$ in this way gives the cyclic code of
length $k_1k_2$ with support $\pi^{-1}(S_1)\cup\{k_1s:s\in S_2\}$, where $\pi$ is reduction
modulo $k_1$, and its distance is $d_1d_2$ (Proposition~\ref{prop:concat} in
Appendix~\ref{app:concat}). Such concatenations appear among the supports found by our search; several
records in Table~\ref{tab:records} are of this form.

\section{Qubits: triply-even cyclic codes}\label{sec:qubit}

Following Bravyi and Haah~\cite{Bravyi_2012}, a binary matrix $G$ with rows $g^{(a)}$ is
\emph{triorthogonal} if $\sum_ig^{(a)}_ig^{(b)}_i\equiv0$ for $a\ne b$ and
$\sum_ig^{(a)}_ig^{(b)}_ig^{(c)}_i\equiv0$ for distinct $a,b,c$; the odd-weight rows, forming
$G_1$, are logical, and the even-weight rows, forming $G_0$, are stabilizers. A binary code
$C$ is \emph{triply even} if all its weights are divisible by $8$, and \emph{classically
triorthogonal} if $\sum_lx_ly_lz_l\equiv0$ for all $x,y,z\in C$. Since $x*x=x$ over $\F_2$,
the cases $x=y=z$ and $x=y$ of this condition say that $C$ has even weights and even pairwise
overlaps. Hence the rowspan of $G_0$ is classically triorthogonal if and only if $G_0$
satisfies the Bravyi--Haah conditions.

Every triply-even code is classically triorthogonal~\cite{Betsumiya_2012}, but not
conversely: the binary repetition code of length $4$ is classically triorthogonal and not
triply even. For cyclic codes of odd length, however, the converse holds. This follows from
McEliece's theorem~\cite{MCELIECE1972}: if $C$ is a nonzero binary cyclic code of odd length
and $t$ is the least integer with $0\in tS$, then all weights of $C$ are divisible by
$2^{t-1}$, and not all by $2^t$.

\begin{theorem}\label{thm:te}
A binary cyclic code $C$ of odd length is triply even if and only if its spectral support
satisfies $0\notin3S$, that is, if and only if $C$ is classically triorthogonal.
\end{theorem}

\begin{proof}
Since $\sum_lx_ly_lz_l=\langle\mathbf 1,x*y*z\rangle$, $C$ is classically triorthogonal if
and only if $C^{\langle3\rangle}$ has even weights, i.e., if and only if $0\notin3S$, by
Lemma~\ref{lem:schur}. A triply-even code is classically
triorthogonal~\cite{Betsumiya_2012}. Conversely, if $0\notin3S$ then also $0\notin S\cup2S$,
since $S\subseteq2S\subseteq3S$, so the least $t$ with $0\in tS$ satisfies $t\ge4$. By
McEliece's theorem, all weights of $C$ are divisible by $2^{t-1}\ge8$.
\end{proof}

We have not found Theorem~\ref{thm:te} stated in the literature. It has two immediate
consequences. First, if $G$ is a triorthogonal matrix of odd length whose stabilizer rows
$G_0$ span a cyclic code, then that code is triply even, since it is classically
triorthogonal by the remark above. Second, every support with $0\notin3S$ gives, via
construction~A, a code on which the transversal $T$ gate acts exactly: if $G_0$ spans a
triply-even code and $G_1=\mathbf 1$ with $n$ odd, every component $\mathbf 1+g$ of
$\ket{\bar1}$ has weight $n-\wt(g)\equiv n\pmod8$, so $T^{\otimes n}$ implements $T^n$ with no
Clifford correction~\cite{Rengaswamy2020,Jain_2025}.

For example, the stabilizer code of the $[[15,1,3]]_2$ code~\cite{MSD} has $S=\{1,2,4,8\}$ up
to a multiplier. No three elements of $S$ sum to $0$ modulo $15$, but $1+2+4+8=15$, so
$t=4$; indeed this is the $[15,4,8]$ simplex code, all of whose nonzero weights equal $8$.

\section{Search strategy}\label{sec:search}

By Propositions~\ref{prop:A} and~\ref{prop:B}, each code is determined by its spectral
support $S$, a union of $p$-cyclotomic cosets modulo $n$. We call $S$ \emph{admissible} for a
construction if it satisfies the corresponding criterion: $0\notin S\cup2S\cup3S$ for
construction~A, and $0\notin3S$ for construction~B.

\paragraph{Enumeration.} For every length $n$ with $p\nmid n$ in the stated ranges, we
enumerate all admissible supports by a depth-first search over unions of cosets, abandoning a
branch as soon as the support becomes inadmissible; since admissibility is inherited by
subsets, this loses nothing. Enlarging $S$ shrinks $C^\perp$ (or $W^\perp$), so it cannot
decrease the distance, nor, by Lemma~\ref{lem:mono}, increase the output error of
distillation. We therefore keep only \emph{maximal} admissible supports, those that cannot be
enlarged by another coset. This is maximality among cyclic codes; we do not consider
non-cyclic extensions. Multipliers $j\mapsto uj$, $u\in\Zn^\times$, permute coordinates and
map $S$ to $uS$, and we keep one support per multiplier orbit; since two cyclic codes can be
permutation-equivalent without being related by a multiplier, our counts are of multiplier
classes and may exceed the number of inequivalent codes.

As an independent check, we repeated the enumeration in Magma~\cite{magma} with a
generalized Bron--Kerbosch algorithm. The usual algorithm finds maximal cliques of a graph,
which would capture only pairwise compatibility of cosets, whereas admissibility involves
sums of three elements. We therefore filter the candidate cosets at every step by the full
admissibility test applied to the current support, and report a support only when no
remaining coset can be added. Because admissibility is inherited by subsets, every maximal
support is reached along some branch, so the search is exhaustive. It finds the same maximal
supports as ours, up to multiplier equivalence and with the same dimensions, at every length:
odd $n\le763$ for $p=2$, $n\le118$ for $p=3$, and $n\le96$ for $p=5,7$.

\paragraph{Distances.} Distances are computed by an exhaustive meet-in-the-middle search,
using cyclic symmetry and scaling to fix $v_0=1$. Every reported value comes with an explicit
minimum-weight vector, checked directly, and the exhaustive search excludes all smaller
weights. Where the search reaches its size limit, it gives a rigorous lower bound; for $p=2$,
a randomized phase may also give an upper bound. Where feasible, we resolved remaining lower
bounds with the Brouwer--Zimmermann algorithm in Magma, which settled $27$ distances,
including that of $[[49,1,7]]_7$. We also use Proposition~\ref{prop:concat}. If a maximal support coincides, up to a
multiplier, with the support of a concatenation of codes of distances $d_1$ and $d_2$, then
$d=d_1d_2$; if it strictly contains such a support, then $d\ge d_1d_2$, since enlarging $S$
cannot decrease the distance. All distances for $p=3$ and $p=5$ are exact; lower bounds remain only for
$p=7$ at $n\in\{80,87,96\}$ and for $p=2$ at $n\ge341$. In the tables, the best distance at a
length is reported as exact only if the distance of every maximal support at that length is
known exactly. Independently, Brouwer--Zimmermann distances agree with ours for all
$n\le225$ when $p=2$, $n\le118$ when $p=3$, and $n\le64$ when $p=5,7$, $884$ codes in total.

\paragraph{Checks.} Every code was rechecked directly against Definition~\ref{def:tri}
(Bravyi--Haah for $p=2$), and every binary code was checked to be triply even. The
transversal gate was verified exhaustively, with exact phases, for $[[15,1,3]]_2$,
$[[19,1,3]]_3$, $[[11,1,3]]_5$ and $[[7,1,3]]_7$. 

\section{Results}\label{sec:results}

We searched every length $n$ with $p\nmid n$ up to $121$ for $p=3$, up to $96$ for $p=5,7$,
and every odd $n$ up to $763$ for $p=2$. Up to multiplier equivalence there are $92$ maximal
supports for construction~A and $195$ for construction~B when $p=3$, $248$ and $395$ when
$p=5$, and $624$ and $1600$ when $p=7$. Table~\ref{tab:records} compares, for each distance,
the shortest code we find with the shortest punctured Reed--Muller code and the shortest
concatenation. Appendix~\ref{app:tables} lists the best code at every length, and
Appendix~\ref{app:supports} gives the spectral supports of selected codes.

\paragraph{$p=3$.} The shortest cyclic distance-$3$ code is the $[[19,1,3]]_3$ code of
Example~\ref{ex:19}, from construction~B; requiring self-orthogonality, i.e.\ $0\notin2S$, as
in~\cite{Krishna_2019}, would raise this to $[[39,1,3]]_3$. Construction~A first reaches
$d\ge3$ at $[[40,1,4]]_3$. Near length $80$ both constructions jump to distance $7$, with
$[[79,1,7]]_3$ (B) and $[[80,1,7]]_3$ (A), and $[[121,1,8]]_3$ (A) is the first code with
$d\ge8$. For $3\le d\le8$ the cyclic codes are shorter than both Reed--Muller codes and
concatenations; at $d=2$ the cyclic code is the Reed--Muller code $[[8,1,2]]_3$. At the
Reed--Muller lengths themselves, the search also improves on Reed--Muller codes. At $n=26$
the only maximal support, up to multiplier equivalence, strictly contains the Reed--Muller
support, giving a $[[26,1,2]]_3$ code with the same distance but a higher threshold
(Table~\ref{tab:perf}). At $n=80$ the search finds $[[80,1,7]]_3$, against $[[80,1,5]]_3$ for
the Reed--Muller code.

\paragraph{$p=5$.} The shortest cyclic distance-$3$ code is $[[11,1,3]]_5$ (B), less than
half the length of the shortest Reed--Muller code with $d\ge3$, $[[24,1,3]]_5$. At larger
distances concatenation does better. The records $[[16,1,4]]_5$ and $[[64,1,8]]_5$ are the
twofold and threefold concatenations of the $[[4,1,2]]_5$ Reed--Solomon code
(Proposition~\ref{prop:concat}), and at $d=5,6$ the concatenation
$[[11,1,3]]_5\circ[[4,1,2]]_5$, of length $44$, is shorter than every cyclic code, including
$[[48,1,5]]_5$. At the Reed--Muller length $24$, construction~A gives $[[24,1,4]]_5$, against
$[[24,1,3]]_5$ for the Reed--Muller code.

\paragraph{$p=7$.} The shortest cyclic codes are $[[7,1,3]]_7$ and $[[28,1,5]]_7$ (B), and
$[[12,1,3]]_7$ and $[[29,1,6]]_7$ (A); for $3\le d\le6$ they are shorter than both
Reed--Muller codes and concatenations. At $d=7$, the $[[49,1,7]]_7$ code (B) has the same
length as the self-concatenation of $[[7,1,3]]_7$, which has distance $9$, and no code with
$d\ge8$ is certified in our range. At the Reed--Muller length $48$, construction~A gives
$[[48,1,6]]_7$, against $[[48,1,4]]_7$ for the Reed--Muller code.

\begin{table}[h]
\centering
\small
\begin{tabular}{@{}c|ccc|ccc|ccc@{}}
\toprule
& \multicolumn{3}{c|}{$p=3$} & \multicolumn{3}{c|}{$p=5$} & \multicolumn{3}{c}{$p=7$}\\
$d$ & RM & conc. & cyclic & RM & conc. & cyclic & RM & conc. & cyclic\\
\midrule
$2$ & $8$ & $64$ & $8$ (A) & $4$ & $16$ & $4$ (A) & $6$ & $36$ & $6$ (A)\\
$3$ & $80$ & $64$ & $19$ (B) & $24$ & $16$ & $11$ (B) & $48$ & $36$ & $7$ (B)\\
$4$ & $80$ & $64$ & $40$ (A, B) & $124$ & $16$ & $16^\dagger$ (A) & $48$ & $36$ & $28$ (B)\\
$5$ & $80$ & $152$ & $79$ (B) & $624$ & $44$ & $48$ (A) & $342$ & $42$ & $28$ (B)\\
$6$ & $242$ & $152$ & $79$ (B) & $624$ & $44$ & $64^\dagger$ (A) & $342$ & $42$ & $29$ (A)\\
$7$ & $242$ & $320$ & $79$ (B) & $624$ & $64$ & $64^\dagger$ (A) & $2400$ & $49$ & $49$ (B)\\
$8$ & $242$ & $320$ & $121$ (A) & $624$ & $64$ & $64^\dagger$ (A) & $2400$ & $49$ & --\\
\bottomrule
\end{tabular}
\caption{Shortest length $n$ of an $[[n,1,d']]_p$ code with $d'\ge d$. ``RM'': punctured generalized
Reed--Muller codes $\mathrm{GRM}_p(r,m)$~\cite{CampbellAnwarBrowne,campbell2014enhanced,Saha:2025frb}
of length $p^m-1$, with the largest admissible $r$ ($3r<m(p-1)$) for each $m$; their distances
follow from~\cite{DGM1970}. ``conc.'': the shortest concatenation of
two or three codes from the other two columns, using $d=d_1d_2$. ``cyclic'': the shortest length at
which our search certifies a cyclic code with distance at least $d$ (construction in parentheses);
$^\dagger$ marks codes that are themselves concatenations of $[[4,1,2]]_5$
(Proposition~\ref{prop:concat}); ``--'': no code with $d'\ge8$ is certified for $p=7$ in the search range. }
\label{tab:records}
\end{table}

\paragraph{$p=2$.} For qubits, both constructions use the supports with $0\notin3S$, which by
Theorem~\ref{thm:te} are exactly the triply-even cyclic codes; we also checked
triple-evenness directly. The main output is therefore a classification. Of the odd
lengths between $3$ and $763$, only $84$ admit a nonzero triply-even cyclic code, and at these lengths
every triply-even cyclic code is, up to multiplier equivalence, a subcode of one of $879$
maximal codes, which we list.

The construction-A records are $[[15,1,3]]_2$, $[[85,1,5]]_2$, $[[127,1,7]]_2$ and
$[[223,1,9]]_2$, all with exact distances; construction~B typically gives codes one shorter
with distance one lower. No support with $n\le763$ is certified to have distance $11$ or
more, although some distances for $n\ge341$ are known only as lower bounds. These codes are
not competitive with the non-cyclic $[[49,1,5]]_2$~\cite{Bravyi_2012},
$[[95,1,7]]_2$~\cite{Sullivan_2024} and $[[185,1,9]]_2$~\cite{Jain_2025} codes.

The search does, however, clarify the role of punctured Reed--Muller codes, whose spectral
support is $\{j:1\le\wt_2(j)\le r\}$ up to a multiplier, where $\wt_2(j)$ is the number of
ones in the binary expansion of $j$~\cite{macwilliams77}. At $n=15$, $31$, $127$ and $255$
this support is maximal, and $\mathrm{RM}(1,4)^*$~\cite{MSD} and $\mathrm{RM}(2,7)^*$ are
the records $[[15,1,3]]_2$ and $[[127,1,7]]_2$. At $n=63$ and $511$, by contrast, the
Reed--Muller support is a proper subset of a maximal support with the same distance. At
$n=63$ the larger code, with a support of size $12$ rather than $6$, has
$\epsilon_{\rm out}\approx21\,\epsilon^3$ and threshold $0.104$ under the noise model of
Section~\ref{sec:perf}, against $651\,\epsilon^3$ and $0.034$ for $\mathrm{RM}(1,6)^*$. The
record $[[85,1,5]]_2$ and $[[223,1,9]]_2$ codes are not of Reed--Muller length.

\section{Distillation performance}\label{sec:perf}

We now evaluate the smallest codes in a single round of distillation. Each noisy input is
modelled as $(1-\epsilon)\ket M\bra M+\frac{\epsilon}{p-1}\sum_{j\ne0}Z^j\ket M\bra MZ^{-j}$
with $\ket M=U\ket+$; this is the twirled form of depolarizing noise of strength
$p_{\rm dep}=\frac p{p-1}\epsilon$~\cite{ACB,CampbellAnwarBrowne}. Measuring the stabilizers
and postselecting accepts an error $Z(j)$ if and only if $j\in\mathcal A^\perp$, and the output
is wrong if and only if $j\notin\mathcal H^\perp$. Hence
\begin{gather*}
P_{\rm acc}=\sum_{j\in\mathcal A^\perp}P(j),\qquad
\epsilon_{\rm out}=\frac1{P_{\rm acc}}\sum_{j\in\mathcal A^\perp\setminus\mathcal H^\perp}P(j),\\
P(j)=(1-\epsilon)^{n-\wt(j)}\Big(\frac\epsilon{p-1}\Big)^{\wt(j)},
\end{gather*}
which we evaluate exactly from the weight distributions of $\mathcal A^\perp$ and
$\mathcal H^\perp$, obtained by the MacWilliams identity~\cite{macwilliams77}. Multiplying $j$
by a nonzero scalar preserves its weight and whether it is accepted, so the output error is
again spread evenly over the $p-1$ nontrivial directions, and the output has the same form as
the input.

Our search keeps only maximal supports (Section~\ref{sec:search}). Under this noise model
that loses nothing, because enlarging a support can only improve the output error:

\begin{lemma}\label{lem:mono}
Consider two codes with the same logical rows and stabilizer spaces
$\mathcal A\subseteq\mathcal A'$. Under the noise model above, for every
$\epsilon\le(p-1)/p$ the larger code has $\epsilon'_{\rm out}\le\epsilon_{\rm out}$. In
particular, its threshold is at least as high.
\end{lemma}

\noindent The proof, given in Appendix~\ref{app:mono}, rewrites the comparison as a
correlation inequality and applies the FKG inequality~\cite{FKG1971}.

Table~\ref{tab:perf} gives, for the smallest new codes and for comparison codes, the overhead
exponent $\gamma=\log n/\log d$~\cite{Bravyi_2012}, which governs the cost of using a code
recursively; the leading behaviour of $\epsilon_{\rm out}$ for small $\epsilon$; the threshold
$\epsilon^*$, below which $\epsilon_{\rm out}<\epsilon$; and the acceptance probability at
$\epsilon=0.01$.

\begin{table}[h]
\centering
\small
\begin{tabular}{@{}lcccc@{}}
\toprule
code & $\gamma$ & $\epsilon_{\rm out}$ for small $\epsilon$ & threshold $\epsilon^*$ & $P_{\rm acc}$ at $\epsilon=0.01$\\
\midrule
$[[19,1,3]]_3$ (B) & $2.68$ & $5.25\,\epsilon^3$ & $0.226$ & $0.83$\\
$[[40,1,4]]_3$ (A) & $2.66$ & $16.25\,\epsilon^4$ & $0.223$ & $0.67$\\
$[[8,1,2]]_3$ (RM) & $3.00$ & $2\,\epsilon^2$ & $0.211$ & $0.92$\\
$[[26,1,2]]_3$ (RM) & $4.70$ & $6.5\,\epsilon^2$ & $0.066$ & $0.77$\\
$[[26,1,2]]_3$ (A) & $4.70$ & $6.5\,\epsilon^2$ & $0.123$ & $0.77$\\
\midrule
$[[11,1,3]]_5$ (B) & $2.18$ & $1.44\,\epsilon^3$ & $0.346$ & $0.90$\\
$[[24,1,3]]_5$ (RM) & $2.89$ & $1.5\,\epsilon^3$ & $0.263$ & $0.79$\\
$[[4,1,2]]_5$ (RS) & $2.00$ & $1.5\,\epsilon^2$ & $0.363$ & $0.96$\\
$[[16,1,4]]_5$ (RS$\circ$RS) & $2.00$ & $3.38\,\epsilon^4$ & $0.363$ & $0.85$\\
\midrule
$[[7,1,3]]_7$ (B) & $1.77$ & $0.97\,\epsilon^3$ & $0.446$ & $0.93$\\
$[[12,1,3]]_7$ (A) & $2.26$ & $1.11\,\epsilon^3$ & $0.366$ & $0.89$\\
$[[6,1,2]]_7$ (RS) & $2.58$ & $2.5\,\epsilon^2$ & $0.232$ & $0.94$\\
\bottomrule
\end{tabular}
\caption{Exact single-round performance under twirled depolarizing noise. RM: punctured Reed--Muller
codes~\cite{CampbellAnwarBrowne,campbell2014enhanced,Saha:2025frb}; RS: Reed--Solomon
codes~\cite{campbell2014enhanced,Krishna_2019}.}
\label{tab:perf}
\end{table}

For $p=3$ and $p=7$ the new codes improve on the shortest Reed--Muller and Reed--Solomon
codes, respectively, in both overhead exponent and threshold. The codes $[[19,1,3]]_3$ and
$[[40,1,4]]_3$ have $\gamma\approx2.7$ and thresholds above $0.22$, against $\gamma=3$ and
$0.211$ for $[[8,1,2]]_3$, and $[[7,1,3]]_7$ has $\gamma=1.77$ and threshold $0.446$, against
$2.58$ and $0.232$ for $[[6,1,2]]_7$. The search also improves on Reed--Muller codes at fixed
length: the cyclic $[[26,1,2]]_3$ code, whose support strictly contains the Reed--Muller
support, has the same leading term but almost twice the threshold. For larger $p$,
Proposition~\ref{prop:mds} gives, for example, $[[13,1,5]]_{13}$ with
$\epsilon_{\rm out}\approx0.062\,\epsilon^5$ and threshold $0.562$, against
$0.286\,\epsilon^4$ and $0.479$ for the Reed--Solomon code $[[12,1,4]]_{13}$. For $p=5$ the
comparison is less favourable: although $[[11,1,3]]_5$ improves on the shortest Reed--Muller
code of the same distance, recursive use of the $[[4,1,2]]_5$ Reed--Solomon code has a
smaller overhead exponent and a higher threshold.

For construction~A with $p\equiv1\pmod3$, the logical gate is $U^\lambda$ with
$\lambda=n\bmod p$. If $\lambda=c^3$ is a cube, $U^\lambda\ket+$ is Clifford-equivalent to
$U\ket+$, via the Clifford gate $\ket j\mapsto\ket{cj}$; if $\lambda$ is not a cube, it is
not~\cite{HowardVala,Bengtsson_2014}. Among the codes in Table~\ref{tab:perf} this affects
only $[[12,1,3]]_7$, with $\lambda=5$, whose output therefore lies in a different Clifford
class from its input. Construction-B codes have $\lambda=-1$, which is a cube, and for
$p=3,5$ every unit is a cube.

Three caveats apply. First, single-round thresholds and $\gamma$ do not by themselves
determine the overhead of a full protocol. Second, Table~\ref{tab:perf} assumes depolarizing
noise. Twirling makes the noise diagonal in the basis $Z^j\ket M$ but does not make the
weights of different $j$ equal, and thresholds for biased noise can differ, as can the
conclusion of Lemma~\ref{lem:mono}, whose proof uses the depolarizing form of the noise.
Third, enlarging a support cannot raise $\epsilon_{\rm out}$, by Lemma~\ref{lem:mono}, but it
does lower the acceptance probability, which also affects the overall overhead.

\section{Conclusion}\label{sec:conclusion}

For cyclic codes of length coprime to $p$, triorthogonality reduces to a sumset condition on
the spectral support. For qubits this gives a complete description of triply-even cyclic
codes of odd length, by the condition $0\notin3S$, and we determine all of them up to length
$763$. For qudits it makes an exhaustive search over two natural constructions feasible. The
resulting codes improve substantially on punctured Reed--Muller codes, in both length and
threshold; for $p=3$, and for $p=7$ at distances $3\le d\le6$, they are also shorter than
concatenations of known codes, while for $p=5$ recursive use of the $[[4,1,2]]_5$
Reed--Solomon code remains better. The criteria also give infinite families directly,
including $[[p,1,2\lfloor p/6\rfloor+1]]_p$ codes for every prime $p\ge7$.

The approach has limitations. We considered two constructions, with unit weights (see
Section~\ref{sec:cyclic}), lengths coprime to $p$, and equivalence under multipliers only, and
some distances are known only as lower bounds. Cyclic codes are also unlikely to be optimal
among all codes: for qubits, non-cyclic constructions are shorter at $d\ge5$, and for qutrits
we have found a better non-cyclic code with $d=3$, to appear in forthcoming work.

Several questions remain open. The smallest distance-$3$ qutrit triorthogonal code is
unknown; a qudit analogue of the classification of~\cite{Nezami_2022,WillsJainSingh2026}
would settle the question. Puncturing cyclic spaces at several points gives codes with
$k>1$ and, with suitable logical rows, codes for multi-qudit gates such as $CCZ$. The sumset
criteria hold verbatim for abelian group codes of order coprime to $p$, and extended cyclic
codes can be treated similarly. They also extend to higher levels of the Clifford
hierarchy: for the gate $\operatorname{diag}(\omega^{j^t})$ with $t<p$, the analogue of
construction~B requires $0\notin tS$, and for qubits a cyclic code of odd length with
$0\notin4S$ has all weights divisible by $16$, by McEliece's theorem. Constacyclic codes may
give further families when $p\equiv1\pmod3$. Finally, it would be interesting to compute
multi-round overheads for the qudit codes found here.

\section*{Acknowledgements}
The author thanks Prof. P.S. Satsangi for discussions on qudit magic state distillation some years
ago, and for earlier work on qudit contextuality~\cite{DPS} that motivated the author's interest in
qudit fault tolerance. The author also thanks T. Saha for sharing his unpublished notes on qudit CSS-T codes, in
which he independently obtained the sufficiency of the cubic conditions without self-orthogonality
for $p\ge5$. This work was supported by DST-SERB grant (CRG/2021/009137).

Large language models (Google Gemini 3.1 and 3.6; Theo Collaborator Beta using GPT-5.5; Claude Sonnet 5.5 and Opus 5.5) assisted with coding, drafting, and revision. I verified all results and code and take full responsibility for this paper.

\section*{Code and data availability}
The search code, a test suite, and the complete results (every maximal admissible spectral support, with
its dimension and certified distance) are available at
\url{https://github.com/spdei/cyclic-triorthogonal-codes}.

\bibliographystyle{quantum2}
\bibliography{qudit2026}

\appendix

\section{Deferred proofs}
\subsection{Spectral description of concatenation}\label{app:concat}

Let $Q_1$ and $Q_2$ be construction-A codes of lengths $k_1$ and $k_2$, with $p\nmid k_1k_2$,
stabilizer codes $C_1$ and $C_2$ with spectral supports $S_1$ and $S_2$, and distances $d_1$
and $d_2$. In their concatenation, with $Q_1$ as the inner code, each of $k_2$ blocks of $k_1$
qudits is encoded with $Q_1$, and the stabilizers of $Q_2$ act through the inner logical
operators. Since the logical row of $Q_1$ is $\mathbf 1$, the stabilizer space $\mathcal A$ of
the concatenation is spanned by the vectors of $C_1$ placed on a single block, together with
the vectors whose restriction to block $b$ is $s_b\mathbf 1$ for some $s\in C_2$. Its logical
row is $\mathbf 1$.

We index $\mathbb Z_{k_1k_2}$ as $l=b+k_2a$ with $a\in\mathbb Z_{k_1}$ and $b\in\mathbb Z_{k_2}$,
so that block $b$ consists of the coordinates $l\equiv b\pmod{k_2}$. Let
$\pi:\mathbb Z_{k_1k_2}\to\mathbb Z_{k_1}$ be reduction modulo $k_1$, and write
$k_1\cdot S_2=\{k_1s:s\in S_2\}$ for the dilation of $S_2$, not to be confused with the
sumset $k_1S_2$.

\begin{proposition}[Concatenation]\label{prop:concat}
The set $S=\pi^{-1}(S_1)\cup k_1\cdot S_2\subseteq\mathbb Z_{k_1k_2}$ is admissible for
construction~A, the resulting code is the concatenation of $Q_1$ with $Q_2$, and its distance
is $d_1d_2$.
\end{proposition}

\begin{proof}
Let $\alpha$ have order $k_1k_2$, and let $\beta=\alpha^{k_2}$ and $\gamma=\alpha^{k_1}$, of
orders $k_1$ and $k_2$. We take $S_1$ and $S_2$ with respect to $\beta$ and $\gamma$; another
choice of root changes them by a multiplier, which gives an equivalent code.

\emph{Step 1: two transforms.} If $v$ is supported on block $b$, with $v_{b+k_2a}=c_a$, then
\[
\widehat v(j)=\sum_a c_a\,\alpha^{j(b+k_2a)}=\alpha^{jb}\,\widehat c(j\bmod k_1),
\]
where $\widehat c$ is computed with $\beta$. If $v$ is constant on blocks, with
$v_l=s_{l\bmod k_2}$, then
\[
\widehat v(j)=\sum_b s_b\,\alpha^{jb}\sum_a\beta^{ja},
\]
which vanishes unless $k_1\mid j$, and $\widehat v(k_1j')=k_1\,\widehat s(j')$, where
$\widehat s$ is computed with $\gamma$.

\emph{Step 2: the stabilizer space is cyclic.} By Step~1, a vector of $C_1$ placed on one
block has spectrum in $\pi^{-1}(S_1)$, and a lifted vector of $C_2$ has spectrum in
$k_1\cdot S_2$. Hence $\mathcal A$ lies in the cyclic code $C$ with support $S$. The two
pieces of $S$ are disjoint, since $k_1\cdot S_2\subseteq\pi^{-1}(0)$ and $0\notin S_1$, so the
two spanning sets of $\mathcal A$ have disjoint spectra and
\[
\dim\mathcal A=k_2|S_1|+|S_2|=|\pi^{-1}(S_1)|+|k_1\cdot S_2|=|S|=\dim C.
\]
Thus $\mathcal A=C$. Both codes have logical row $\mathbf 1$, so by Lemma~\ref{lem:css} the
concatenation is the construction-A code with support $S$.

\emph{Step 3: admissibility.} Suppose a sum of one, two or three elements of $S$ vanishes. If
at least one of them lies in $\pi^{-1}(S_1)$, reduce modulo $k_1$: the elements of
$k_1\cdot S_2$ become $0$, leaving a vanishing sum of one to three elements of $S_1$. Otherwise
all of them lie in $k_1\cdot S_2$, and dividing by $k_1$ gives a vanishing sum of elements of
$S_2$ modulo $k_2$. Both cases are excluded by the admissibility of $S_1$ and $S_2$.

\emph{Step 4: the distance.} For $v\in\F_p^{k_1k_2}$, let $v^{(b)}\in\F_p^{k_1}$ be its
restriction to block $b$ and $\sigma_b=\sum_av^{(b)}_a$ its block sums. Orthogonality to the
stabilizers on block $b$ says $v^{(b)}\in C_1^\perp$, and orthogonality to a lifted vector
$s\in C_2$ says $\sum_bs_b\sigma_b=0$. Hence
\[
C^\perp=\{v:\ v^{(b)}\in C_1^\perp\text{ for all }b,\ \sigma\in C_2^\perp\},
\qquad\textstyle\sum_lv_l=\sum_b\sigma_b .
\]
By Proposition~\ref{prop:A}, $d$ is the least weight of $v\in C^\perp$ with
$\sum_lv_l\neq0$. For such $v$, $\sigma\in C_2^\perp$ has nonzero sum, so $\sigma_b\neq0$ for
at least $d_2$ blocks; on each of them $v^{(b)}\in C_1^\perp$ has nonzero sum and hence weight at
least $d_1$. So $d\ge d_1d_2$. Conversely, take $u\in C_1^\perp$ of weight $d_1$ with sum
$t\ne0$, and $\sigma\in C_2^\perp$ of weight $d_2$ with nonzero sum. Setting
$v^{(b)}=(\sigma_b/t)\,u$ gives $v\in C^\perp$ with block sums $\sigma$, nonzero total sum, and
weight $d_1d_2$.
\end{proof}

Concatenations involving construction B are not considered here, since construction-B codes
are not cyclic.

\subsection{Proof of Lemma~\ref{lem:mono}}\label{app:mono}
Here we prove Lemma~\ref{lem:mono} from the main text. We use the notation of Section~\ref{sec:perf}.

\begin{proof}
Let $V=\mathcal A^\perp$ and $V'=\mathcal A'^\perp\subseteq V$ be the sets of accepted
errors, and let $L$ be the set of vectors orthogonal to every logical row, so that the
accepted errors that act trivially on the output are $V\cap L$ and $V'\cap L$. For
$R\subseteq\{1,\dots,n\}$ and a subspace $Y$, let $Y_R$ be the subspace of vectors in $Y$
supported on $R$.

\emph{Step 1: resampling the noise.} Let $\eta=p\epsilon/(p-1)\le1$. The noise is
equivalent to the following: choose a random set $R$ containing each coordinate
independently with probability $\eta$, and apply $Z(j)$ with $j$ uniform among vectors
supported on $R$. Each coordinate of $j$ is then nonzero with probability
$\eta(p-1)/p=\epsilon$, with all nonzero values equally likely.

\emph{Step 2: reduction to a correlation inequality.} Given $R$, the smaller code accepts
with probability $|V_R|/p^{|R|}$, and conditioned on acceptance $j$ is uniform on $V_R$.
Hence, conditioned on acceptance, $R$ has distribution
\[
\mu(R)\propto\eta^{|R|}(1-\eta)^{n-|R|}\,p^{\dim V_R-|R|},
\]
and
\[
P(j\in L\mid j\in V)=\mathbb E_\mu[f],\quad
P(j\in V'\mid j\in V)=\mathbb E_\mu[g],\quad
P(j\in V'\cap L\mid j\in V)=\mathbb E_\mu[h],
\]
where $f(R)=|V_R\cap L|/|V_R|$, $g(R)=|V_R\cap V'|/|V_R|$ and
$h(R)=|V_R\cap V'\cap L|/|V_R|$. Since $1-\epsilon_{\rm out}=\mathbb E_\mu[f]$ and
$1-\epsilon'_{\rm out}=P(j\in L\mid j\in V')=\mathbb E_\mu[h]/\mathbb E_\mu[g]$, the lemma
is equivalent to
\begin{equation}\label{eq:corr}
\mathbb E_\mu[h]\ge\mathbb E_\mu[f]\,\mathbb E_\mu[g].
\end{equation}

\emph{Step 3: $h\ge fg$ pointwise.} For subspaces $A,B$ of $V_R$ we have
$|A|\,|B|=|A\cap B|\,|A+B|\le|A\cap B|\,|V_R|$. Taking $A=V_R\cap L$ and $B=V_R\cap V'$
gives $h(R)\ge f(R)g(R)$.

\emph{Step 4: $f$ and $g$ are decreasing.} $1/f(R)=|V_R/(V_R\cap L)|$ is the size of the
image of $V_R$ in $V/(V\cap L)$, which can only grow as $R$ grows. The same argument
applies to $g$, with $V'$ in place of $L$.

\emph{Step 5: the FKG inequality.} The measure $\mu$ satisfies the lattice condition
$\mu(R\cup R')\,\mu(R\cap R')\ge\mu(R)\,\mu(R')$. The factors
$\eta^{|R|}(1-\eta)^{n-|R|}$ and $p^{-|R|}$ are modular in $R$, and since
$V_{R\cap R'}=V_R\cap V_{R'}$ and $V_{R\cup R'}\supseteq V_R+V_{R'}$,
\[
\dim V_{R\cup R'}+\dim V_{R\cap R'}\ge\dim V_R+\dim V_{R'}.
\]
By the FKG inequality~\cite{FKG1971}, the decreasing functions $f$ and $g$ are positively
correlated under $\mu$, so $\mathbb E_\mu[fg]\ge\mathbb E_\mu[f]\,\mathbb E_\mu[g]$.
Together with Step 3 this gives~\eqref{eq:corr}. (For $\eta=1$, $\mu$ is a point mass and
\eqref{eq:corr} follows from Step 3 alone.)
\end{proof}

% Generated by scripts/make_tables.py from results/; do not edit by hand.
\section{Best code at each length}\label{app:tables}
For each construction and length, we list the code with the largest distance among all maximal supports, followed in parentheses by the number of maximal supports up to multiplier equivalence. Lengths at which every code has distance $1$ are omitted. An entry is exact only if the distance of every maximal support at that length is known exactly; otherwise $\ge d$ gives the best rigorous lower bound, and $a\text{--}b$ means that the best distance lies between $a$ and $b$. Lower bounds also use Proposition~\ref{prop:concat}: a support containing the concatenation of two construction-A supports, of distances $d_1$ and $d_2$, has distance at least $d_1d_2$, with equality if it equals that support.

\paragraph{Qubits.} Both constructions use the same supports, all of them triply even. Each entry gives the construction-A code $[[n,1,d]]$ and, after ``B:'', the distance of the construction-B code of length $n-1$.
\begin{center}\small
\begin{tabular}{lll}
$[[15,1,3]]$, B: $2$ (1) & $[[31,1,3]]$, B: $2$ (1) & $[[45,1,3]]$, B: $2$ (2) \\
$[[51,1,3]]$, B: $2$ (1) & $[[63,1,3]]$, B: $2$ (1) & $[[75,1,3]]$, B: $2$ (2) \\
$[[85,1,5]]$, B: $4$ (1) & $[[89,1,5]]$, B: $3$ (1) & $[[93,1,5]]$, B: $4$ (3) \\
$[[105,1,3]]$, B: $2$ (2) & $[[117,1,3]]$, B: $2$ (1) & $[[123,1,3]]$, B: $2$ (1) \\
$[[127,1,7]]$, B: $6$ (3) & $[[135,1,3]]$, B: $2$ (3) & $[[153,1,3]]$, B: $2$ (2) \\
$[[155,1,5]]$, B: $4$ (3) & $[[165,1,3]]$, B: $2$ (2) & $[[189,1,5]]$, B: $4$ (5) \\
$[[195,1,3]]$, B: $2$ (3) & $[[215,1,5]]$, B: $4$ (1) & $[[217,1,5]]$, B: $4$ (6) \\
$[[221,1,5]]$, B: $4$ (1) & $[[223,1,9]]$, B: $8$ (1) & $[[225,1,9]]$, B: $3$ (4) \\
$[[233,1,7]]$, B: $6$ (1) & $[[255,1,7]]$, B: $6$ (46) & $[[267,1,7]]$, B: $6$ (3) \\
$[[275,1,5]]$, B: $4$ (1) & $[[279,1,7]]$, B: $5$ (6) & $[[285,1,3]]$, B: $2$ (2) \\
$[[315,1,7]]$, B: $6$ (29) & $[[339,1,3]]$, B: $2$ (1) & $[[341,1,\ge9]]$, B: $\ge7$ (31) \\
$[[345,1,3]]$, B: $2$ (2) & $[[351,1,7]]$, B: $6$ (3) & $[[357,1,3]]$, B: $2$ (4) \\
$[[369,1,3]]$, B: $2$ (2) & $[[375,1,3]]$, B: $2$ (3) & $[[381,1,\ge7]]$, B: $6\text{--}7$ (41) \\
$[[399,1,3]]$, B: $2$ (2) & $[[403,1,5]]$, B: $4$ (3) & $[[405,1,3]]$, B: $2$ (4) \\
$[[425,1,5]]$, B: $4$ (6) & $[[431,1,\ge7]]$, B: $\ge6$ (1) & $[[435,1,3]]$, B: $2$ (2) \\
$[[441,1,3]]$, B: $2$ (7) & $[[445,1,7]]$, B: $6$ (3) & $[[451,1,\ge7]]$, B: $6\text{--}7$ (2) \\
$[[455,1,5]]$, B: $4$ (5) & $[[459,1,3]]$, B: $2$ (3) & $[[465,1,\ge9]]$, B: $6\text{--}7$ (23) \\
$[[471,1,3]]$, B: $2$ (1) & $[[485,1,5]]$, B: $4$ (1) & $[[495,1,5]]$, B: $3$ (6) \\
$[[511,1,\ge7]]$, B: $\ge6$ (158) & $[[525,1,3]]$, B: $2$ (8) & $[[527,1,5]]$, B: $4$ (3) \\
$[[549,1,3]]$, B: $2$ (1) & $[[555,1,3]]$, B: $2$ (2) & $[[561,1,3]]$, B: $2$ (4) \\
$[[567,1,5]]$, B: $4$ (12) & $[[585,1,\ge7]]$, B: $\ge6$ (189) & $[[589,1,5]]$, B: $4$ (3) \\
$[[595,1,5]]$, B: $4$ (12) & $[[601,1,\ge7]]$, B: $\ge6$ (2) & $[[615,1,3]]$, B: $2$ (9) \\
$[[623,1,\ge7]]$, B: $6$ (8) & $[[635,1,\ge7]]$, B: $\ge6$ (41) & $[[645,1,5]]$, B: $4$ (5) \\
$[[651,1,\ge7]]$, B: $\ge6$ (64) & $[[657,1,3]]$, B: $2$ (12) & $[[663,1,\ge7]]$, B: $\ge6$ (9) \\
$[[669,1,\ge7]]$, B: $\ge6$ (1) & $[[671,1,\ge7]]$, B: $\ge6$ (1) & $[[675,1,\ge9]]$, B: $3$ (12) \\
$[[693,1,3]]$, B: $2$ (9) & $[[697,1,5]]$, B: $4$ (1) & $[[699,1,7]]$, B: $5\text{--}6$ (1) \\
$[[705,1,3]]$, B: $2$ (2) & $[[713,1,5]]$, B: $4$ (6) & $[[731,1,5]]$, B: $4$ (1) \\
$[[735,1,3]]$, B: $2$ (6) & $[[741,1,3]]$, B: $2$ (1) & $[[763,1,3]]$, B: $2$ (1) \\
\end{tabular}\end{center}

\paragraph{Qutrits, construction A.} 92 maximal supports in total.
\begin{center}\small
\begin{tabular}{llll}
$[[8,1,2]]$ (1) & $[[16,1,2]]$ (2) & $[[20,1,2]]$ (1) & $[[22,1,2]]$ (1) \\
$[[26,1,2]]$ (1) & $[[32,1,2]]$ (3) & $[[40,1,4]]$ (4) & $[[44,1,2]]$ (2) \\
$[[46,1,2]]$ (1) & $[[52,1,4]]$ (2) & $[[56,1,2]]$ (3) & $[[64,1,4]]$ (4) \\
$[[68,1,2]]$ (1) & $[[70,1,2]]$ (1) & $[[80,1,7]]$ (18) & $[[88,1,4]]$ (5) \\
$[[92,1,2]]$ (2) & $[[94,1,2]]$ (1) & $[[100,1,2]]$ (2) & $[[104,1,5]]$ (16) \\
$[[110,1,2]]$ (2) & $[[112,1,4]]$ (11) & $[[116,1,2]]$ (1) & $[[118,1,2]]$ (1) \\
$[[121,1,8]]$ (6) &  &  &  \\
\end{tabular}\end{center}

\paragraph{Qutrits, construction B.} 195 maximal supports in total.
\begin{center}\small
\begin{tabular}{llll}
$[[15,1,2]]$ (4) & $[[19,1,3]]$ (3) & $[[27,1,3]]$ (4) & $[[31,1,2]]$ (6) \\
$[[39,1,3]]$ (9) & $[[40,1,4]]$ (1) & $[[43,1,3]]$ (4) & $[[51,1,3]]$ (6) \\
$[[55,1,3]]$ (8) & $[[60,1,4]]$ (1) & $[[63,1,2]]$ (9) & $[[67,1,3]]$ (3) \\
$[[75,1,3]]$ (4) & $[[79,1,7]]$ (33) & $[[81,1,7]]$ (3) & $[[87,1,3]]$ (10) \\
$[[91,1,3]]$ (4) & $[[99,1,3]]$ (4) & $[[103,1,5]]$ (25) & $[[111,1,4]]$ (15) \\
$[[115,1,3]]$ (3) & $[[120,1,7]]$ (13) &  &  \\
\end{tabular}\end{center}

\paragraph{Ququints, construction A.} 248 maximal supports in total.
\begin{center}\small
\begin{tabular}{llll}
$[[4,1,2]]$ (1) & $[[8,1,2]]$ (2) & $[[12,1,2]]$ (2) & $[[16,1,4]]$ (3) \\
$[[22,1,2]]$ (1) & $[[24,1,4]]$ (7) & $[[28,1,2]]$ (2) & $[[32,1,4]]$ (6) \\
$[[36,1,2]]$ (4) & $[[38,1,2]]$ (1) & $[[39,1,3]]$ (1) & $[[44,1,3]]$ (4) \\
$[[48,1,5]]$ (23) & $[[52,1,2]]$ (4) & $[[56,1,3]]$ (8) & $[[62,1,2]]$ (4) \\
$[[64,1,8]]$ (10) & $[[66,1,2]]$ (2) & $[[68,1,2]]$ (2) & $[[71,1,6]]$ (3) \\
$[[72,1,4]]$ (30) & $[[76,1,3]]$ (4) & $[[78,1,3]]$ (5) & $[[84,1,2]]$ (12) \\
$[[88,1,4]]$ (15) & $[[91,1,4]]$ (1) & $[[92,1,2]]$ (2) & $[[96,1,8]]$ (89) \\
\end{tabular}\end{center}

\paragraph{Ququints, construction B.} 395 maximal supports in total.
\begin{center}\small
\begin{tabular}{llll}
$[[7,1,2]]$ (3) & $[[11,1,3]]$ (3) & $[[12,1,3]]$ (1) & $[[15,1,2]]$ (5) \\
$[[23,1,3]]$ (12) & $[[25,1,3]]$ (3) & $[[27,1,3]]$ (3) & $[[31,1,2]]$ (8) \\
$[[35,1,3]]$ (4) & $[[38,1,3]]$ (3) & $[[43,1,3]]$ (5) & $[[47,1,4]]$ (40) \\
$[[51,1,3]]$ (9) & $[[55,1,3]]$ (9) & $[[62,1,2]]$ (1) & $[[63,1,2]]$ (14) \\
$[[67,1,3]]$ (3) & $[[70,1,5]]$ (4) & $[[71,1,3]]$ (43) & $[[75,1,3]]$ (5) \\
$[[77,1,5]]$ (36) & $[[83,1,3]]$ (9) & $[[86,1,2]]$ (1) & $[[87,1,3]]$ (14) \\
$[[90,1,6]]$ (2) & $[[91,1,3]]$ (3) & $[[95,1,4]]$ (134) &  \\
\end{tabular}\end{center}

\paragraph{$p=7$, construction A.} 624 maximal supports in total.
\begin{center}\small
\begin{tabular}{llll}
$[[6,1,2]]$ (1) & $[[12,1,3]]$ (1) & $[[16,1,2]]$ (1) & $[[18,1,2]]$ (2) \\
$[[20,1,2]]$ (1) & $[[24,1,3]]$ (4) & $[[29,1,6]]$ (1) & $[[30,1,2]]$ (2) \\
$[[32,1,3]]$ (4) & $[[36,1,4]]$ (6) & $[[38,1,2]]$ (2) & $[[40,1,4]]$ (3) \\
$[[48,1,6]]$ (46) & $[[52,1,2]]$ (1) & $[[54,1,2]]$ (4) & $[[58,1,6]]$ (2) \\
$[[60,1,4]]$ (9) & $[[62,1,2]]$ (1) & $[[64,1,5]]$ (10) & $[[66,1,2]]$ (2) \\
$[[68,1,2]]$ (1) & $[[72,1,7]]$ (43) & $[[74,1,2]]$ (1) & $[[75,1,3]]$ (2) \\
$[[76,1,4]]$ (6) & $[[78,1,2]]$ (2) & $[[80,1,\ge7]]$ (28) & $[[87,1,\ge6]]$ (5) \\
$[[90,1,2]]$ (8) & $[[94,1,2]]$ (1) & $[[96,1,\ge6]]$ (424) &  \\
\end{tabular}\end{center}

\paragraph{$p=7$, construction B.} 1600 maximal supports in total.
\begin{center}\small
\begin{tabular}{llll}
$[[7,1,3]]$ (3) & $[[11,1,3]]$ (4) & $[[15,1,3]]$ (5) & $[[19,1,3]]$ (3) \\
$[[23,1,3]]$ (20) & $[[24,1,3]]$ (1) & $[[28,1,5]]$ (1) & $[[31,1,4]]$ (10) \\
$[[35,1,3]]$ (9) & $[[39,1,3]]$ (9) & $[[43,1,3]]$ (4) & $[[47,1,5]]$ (119) \\
$[[49,1,7]]$ (4) & $[[51,1,3]]$ (3) & $[[57,1,5]]$ (2) & $[[59,1,3]]$ (16) \\
$[[63,1,4]]$ (18) & $[[64,1,4]]$ (1) & $[[67,1,3]]$ (3) & $[[71,1,5]]$ (165) \\
$[[74,1,4]]$ (8) & $[[75,1,3]]$ (9) & $[[79,1,\ge7]]$ (51) & $[[86,1,5]]$ (6) \\
$[[87,1,5]]$ (10) & $[[91,1,3]]$ (4) & $[[95,1,\ge6]]$ (1092) &  \\
\end{tabular}\end{center}

\section{Spectral supports of selected codes}\label{app:supports}

Supports are given with respect to the primitive root $\alpha$ used in the data files: GAP's
$Z(p^m)^{(p^m-1)/n}$, or for $n=64$, $p=5$ a root of the first factor of $\Phi_{64}$ returned by
PARI/GP. Another choice of $\alpha$ changes $S$ by a multiplier, which gives an equivalent code. 

\begin{center}
\small
\begin{tabular}{@{}ll>{\raggedright\arraybackslash}p{7.2cm}@{}}
\toprule
code & constr., length & spectral support $S$\\
\midrule
$[[15,1,3]]_2$ & A, $15$ & $\{1,\allowbreak 2,\allowbreak 4,\allowbreak 8\}$\\
$[[85,1,5]]_2$ & A, $85$ & $\{1,\allowbreak 2,\allowbreak 3,\allowbreak 4,\allowbreak 6,\allowbreak 8,\allowbreak 11,\allowbreak 12,\allowbreak 16,\allowbreak 22,\allowbreak 24,\allowbreak 32,\allowbreak 43,\allowbreak 44,\allowbreak 48,\allowbreak 64\}$\\
$[[127,1,7]]_2$ & A, $127$ & $\{1,\allowbreak 2,\allowbreak 3,\allowbreak 4,\allowbreak 5,\allowbreak 6,\allowbreak 8,\allowbreak 9,\allowbreak 10,\allowbreak 12,\allowbreak 16,\allowbreak 17,\allowbreak 18,\allowbreak 20,\allowbreak 24,\allowbreak 32,\allowbreak 33,\allowbreak 34,\allowbreak 36,\allowbreak 40,\allowbreak 48,\allowbreak 64,\allowbreak 65,\allowbreak 66,\allowbreak 68,\allowbreak 72,\allowbreak 80,\allowbreak 96\}$\\
$[[223,1,9]]_2$ & A, $223$ & $\{1,\allowbreak 2,\allowbreak 4,\allowbreak 7,\allowbreak 8,\allowbreak 14,\allowbreak 15,\allowbreak 16,\allowbreak 17,\allowbreak 28,\allowbreak 30,\allowbreak 32,\allowbreak 33,\allowbreak 34,\allowbreak 41,\allowbreak 49,\allowbreak 56,\allowbreak 60,\allowbreak 64,\allowbreak 66,\allowbreak 68,\allowbreak 82,\allowbreak 98,\allowbreak 105,\allowbreak 112,\allowbreak 115,\allowbreak 119,\allowbreak 120,\allowbreak 128,\allowbreak 132,\allowbreak 136,\allowbreak 164,\allowbreak 169,\allowbreak 171,\allowbreak 196,\allowbreak 197,\allowbreak 210\}$\\
\midrule
$[[8,1,2]]_3$ & A, $8$ & $\{1,\allowbreak 3\}$\\
$[[15,1,2]]_3$ & B, $16$ & $\{1,\allowbreak 3,\allowbreak 8,\allowbreak 9,\allowbreak 11\}$\\
$[[19,1,3]]_3$ & B, $20$ & $\{2,\allowbreak 5,\allowbreak 6,\allowbreak 14,\allowbreak 15,\allowbreak 18\}$\\
$[[40,1,4]]_3$ & A, $40$ & $\{1,\allowbreak 2,\allowbreak 3,\allowbreak 5,\allowbreak 6,\allowbreak 9,\allowbreak 14,\allowbreak 15,\allowbreak 18,\allowbreak 27\}$\\
$[[40,1,4]]_3$ & B, $41$ & $\{1,\allowbreak 3,\allowbreak 9,\allowbreak 14,\allowbreak 27,\allowbreak 32,\allowbreak 38,\allowbreak 40\}$\\
$[[80,1,7]]_3$ & A, $80$ & $\{1,\allowbreak 2,\allowbreak 3,\allowbreak 4,\allowbreak 5,\allowbreak 6,\allowbreak 9,\allowbreak 12,\allowbreak 15,\allowbreak 18,\allowbreak 27,\allowbreak 28,\allowbreak 36,\allowbreak 45,\allowbreak 54,\allowbreak 55\}$\\
$[[79,1,7]]_3$ & B, $80$ & $\{1,\allowbreak 3,\allowbreak 9,\allowbreak 14,\allowbreak 22,\allowbreak 27,\allowbreak 34,\allowbreak 38,\allowbreak 40,\allowbreak 42,\allowbreak 46,\allowbreak 53,\allowbreak 58,\allowbreak 66,\allowbreak 71,\allowbreak 77,\allowbreak 79\}$\\
$[[121,1,8]]_3$ & A, $121$ & $\{1,\allowbreak 2,\allowbreak 3,\allowbreak 6,\allowbreak 7,\allowbreak 9,\allowbreak 16,\allowbreak 18,\allowbreak 21,\allowbreak 23,\allowbreak 27,\allowbreak 41,\allowbreak 48,\allowbreak 54,\allowbreak 63,\allowbreak 68,\allowbreak 69,\allowbreak 81,\allowbreak 83,\allowbreak 86\}$\\
\midrule
$[[4,1,2]]_5$ & A, $4$ & $\{1\}$\\
$[[11,1,3]]_5$ & B, $12$ & $\{2,\allowbreak 3,\allowbreak 9,\allowbreak 10\}$\\
$[[16,1,4]]_5$ & A, $16$ & $\{1,\allowbreak 4,\allowbreak 5,\allowbreak 9,\allowbreak 13\}$\\
$[[47,1,4]]_5$ & B, $48$ & $\{1,\allowbreak 2,\allowbreak 3,\allowbreak 5,\allowbreak 10,\allowbreak 12,\allowbreak 15,\allowbreak 25,\allowbreak 27,\allowbreak 29,\allowbreak 39\}$\\
$[[48,1,5]]_5$ & A, $48$ & $\{1,\allowbreak 2,\allowbreak 3,\allowbreak 5,\allowbreak 10,\allowbreak 12,\allowbreak 15,\allowbreak 25,\allowbreak 27,\allowbreak 29,\allowbreak 39\}$\\
$[[70,1,5]]_5$ & B, $71$ & $\{1,\allowbreak 5,\allowbreak 14,\allowbreak 17,\allowbreak 25,\allowbreak 46,\allowbreak 54,\allowbreak 57,\allowbreak 66,\allowbreak 70\}$\\
$[[90,1,6]]_5$ & B, $91$ & $\{1,\allowbreak 5,\allowbreak 25,\allowbreak 28,\allowbreak 31,\allowbreak 34,\allowbreak 42,\allowbreak 47,\allowbreak 49,\allowbreak 51,\allowbreak 53,\allowbreak 63,\allowbreak 64,\allowbreak 73,\allowbreak 79,\allowbreak 83\}$\\
$[[64,1,8]]_5$ & A, $64$ & $\{1,\allowbreak 4,\allowbreak 5,\allowbreak 9,\allowbreak 13,\allowbreak 16,\allowbreak 17,\allowbreak 20,\allowbreak 21,\allowbreak 25,\allowbreak 29,\allowbreak 33,\allowbreak 36,\allowbreak 37,\allowbreak 41,\allowbreak 45,\allowbreak 49,\allowbreak 52,\allowbreak 53,\allowbreak 57,\allowbreak 61\}$\\
\midrule
$[[6,1,2]]_7$ & A, $6$ & $\{1\}$\\
$[[7,1,3]]_7$ & B, $8$ & $\{1,\allowbreak 4,\allowbreak 7\}$\\
$[[12,1,3]]_7$ & A, $12$ & $\{1,\allowbreak 2,\allowbreak 7\}$\\
$[[28,1,5]]_7$ & B, $29$ & $\{1,\allowbreak 7,\allowbreak 16,\allowbreak 20,\allowbreak 23,\allowbreak 24,\allowbreak 25\}$\\
$[[29,1,6]]_7$ & A, $29$ & $\{1,\allowbreak 7,\allowbreak 16,\allowbreak 20,\allowbreak 23,\allowbreak 24,\allowbreak 25\}$\\
$[[49,1,7]]_7$ & B, $50$ & $\{1,\allowbreak 3,\allowbreak 7,\allowbreak 12,\allowbreak 16,\allowbreak 21,\allowbreak 25,\allowbreak 29,\allowbreak 34,\allowbreak 38,\allowbreak 43,\allowbreak 47,\allowbreak 49\}$\\
\bottomrule
\end{tabular}
\end{center}

\end{document}